\documentclass{amsart}

\usepackage{cite}
\usepackage{amsmath,amssymb,amsthm}
\usepackage{amsfonts}
\usepackage{mathrsfs}
\usepackage[bookmarks=true,%
    colorlinks=true,%
    linkcolor=blue,%
    citecolor=blue,%
    filecolor=blue,%
    menucolor=blue,%
    urlcolor=blue,%
    breaklinks=true]{hyperref}
\usepackage{color}

\usepackage{graphicx}
\usepackage[small,bf%,nooneline
]{caption}
\newtheorem{theorem}{Theorem}
\newtheorem{lemma}{Lemma}
\newtheorem{definition-proposition}{Definition-Proposition}
\theoremstyle{remark}
\newtheorem{definition}{Definition}
\newtheorem{remark}{Remark}
\newcommand{\ds}{\displaystyle}

\def\EXP{\operatorname{{e}}}

\newcommand{\uop}{\boldsymbol{u}}
\newcommand{\vop}{\boldsymbol{v}}

\newcommand{\pop}{\boldsymbol{p}}
\newcommand{\xop}{\boldsymbol{x}}

\newcommand{\ii}{\mathsf{i}}
\newcommand{\bb}{\mathsf{b}}
\newcommand{\bP}{\mathbb{P}}
\newcommand{\bF}{\mathbb{F}}
\newcommand{\R}{\mathbb{R}}
\newcommand{\Z}{\mathbb{Z}}
\newcommand{\C}{\mathbb{C}}

\newcommand{\energy}{\varepsilon}

\newcommand{\bchi}{\boldsymbol{\chi}}

\title{Spectral equations for the modular oscillator}
\author{Rinat M.~Kashaev}
\address{Section de math\'ematiques, Universit\'e de Gen\`eve,
2-4 rue du Li\`evre, 1211 Gen\`eve 4, Suisse\\}
\email{Rinat.Kashaev@unige.ch}
\author{Sergey M.~Sergeev}
\address{Faculty of Education Science Technology \& Mathematics\\
University of Canberra, Bruce ACT 2601, Australia.}
\address{
Department of Theoretical Physics,
Research School of Physics and Engineering\\
Australian National University, Canberra, ACT 0200, Australia.}
\email{Sergey.Sergeev@canberra.edu.au}
\date{March 22, 2017}
\thanks{The work is partially supported by Australian Research Council, and the work of the first-named author is supported in part by Swiss National Science Foundation.}
 \dedicatory{In memory of Ludwig Faddeev}

\begin{document}

%\vglue 2cm

%\begin{center}

%\title{On modular spectral equations. Modular oscillator.}

%\vspace{.5cm}

%\author{Rinat M.~Kashaev$^{1}$ and Sergey M.~Sergeev$^{2,3}$}

%\vspace{.5cm}

%\address{$^1$ Section de Math\'ematiques
%Universit\'e de Gen\`eve,\\ 
%1211 Gen\`eve 4, Suisse.\\\ \\
%$^2$Faculty of Education Science Technology \& Mathematics,\\
%University of Canberra, Bruce ACT 2601, Australia.\\\ \\
%$^2$Department of Theoretical Physics,
%Research School of Physics and Engineering,\\
%Australian National University, Canberra, ACT 0200, Australia.}

%\abstract{Abstract. We discuss the stationary Schr\"odinger equations for a pair of modular conjugated Hamiltonians... Intensive numerical studies.}

%\end{center}

%------------------------------------------------------
%\vspace{1cm}
%\newpage

\begin{abstract} 
 Motivated by applications for non-perturbative topological strings in toric Calabi--Yau manifolds, we discuss the spectral problem  for a pair of commuting modular conjugate (in the sense of Faddeev) Harper type operators, corresponding to a special case of the quantized mirror curve  of local $\bP^1\times\bP^1$ and complex values of Planck's constant. We illustrate our analytical results by numerical calculations. 
\end{abstract}
\maketitle
\section{Introduction}
Topological string theory is known for its prominent role in relating different subjects of mathematical physics and mathematics such as Chern--Simons theory, Gromov--Witten invariants, mirror symmetry, enumerative geometry, quantum topology,
integrable systems, supersymmetric gauge theory, random matrix theory, etc.  The recent progress in string theory has lead to remarkable connections between spectral theory, integrable systems and local mirror symmetry. The idea of relating topological strings in toric Calabi--Yau manifolds with quantum mechanical spectral problems originating from integrable systems was put forward in \cite{MR2191887} and subsequently has been materialised to a powerful and numerically testable conjecture of Grassi--Hatsuda--Mari\~no~\cite{MR3556519}.  It has been shown that quantization of mirror curves leads to trace class (one-dimensional) quantum mechanical operators, whose spectral properties not only contain the full information of the enumerative geometry of the underlying Calabi--Yau manifold, but also allow to formulate the string theory non-perturbatively, see \cite{Marino2015} for a review and references therein.  Outside the context of string theory, these results are preceded by connections of some quantum mechanical spectral problems with integrable systems and conformal field theory~\cite{MR1733841,MR1832065}.

In the case of toric Calabi--Yau threefold known as local $\bP^1\times\bP^1$ or local $\bF_0$, the corresponding operator obtained from quantisation of the mirror curve is of the form
\begin{equation}\label{Hams0}
\boldsymbol{O}_{\bF_0,m}=\vop+\vop^{-1} + \uop + m\uop^{-1}
\end{equation}
with positive self-adjoint operators $\uop$ and $\vop$ satisfying the Heisenberg--Weyl commutation relation
\begin{equation}\label{hwcr}
\uop\vop=\EXP^{\ii\hbar}\vop\uop,\quad \hbar\in\R_{>0},
\end{equation}
and positive real parameter $m$.
In the special case $m=1$, operator $\boldsymbol{O}_{\bF_0,1}$ structurally resembles the Harper  operator \cite{0370-1298-68-10-304}, but the fact that operators $\uop$ and $\vop$ are positive self-adjoint rather than unitary makes its quantum mechanical content very different: it has purely discrete positive spectrum, and its inverse $\boldsymbol{O}_{\bF_0,m}^{-1}$ is of trace class~\cite{MR3537342,MR3546986}. Nonetheless, recent findings in \cite{HatsudaKatsuraTachikawa2016,HatsudaSugimotoXu2017} indicate to hidden connections of the two cases through Faddeev's modular duality~\cite{MR1345554}.  

The spectral problem of operator \eqref{Hams0} has been addressed in \cite{MR2730782} as a part of a larger class of integrable systems related to supersymmetric gauge theories, and, more recently, in \cite{ Kashani-Poor2016} from the perspective of exact WKB approximation and in \cite{MarinoZakany2016} by using a matrix integral representation of the eigenfunctions. 
One can approach the same problem from the standpoint of quantum integrable systems by noting that both the Bethe ansatz equations and the quantum method of separation of variables can be formulated in more general and universal framework of Baxter's $T-Q$ functional difference equations~\cite{MR690578} which, in their turn, can be interpreted as (one-dimensional) quantum mechanical spectral problems. Complemented with Faddeev's modular duality idea, the second-named author has shown in \cite{MR2165903} that the modular $T-Q$ equations of general form admit a particularly tractable approach to their solution in the so called ``strong coupling'' regime corresponding to a complexified parameter $\hbar$ in the commutation relation~\eqref{hwcr}, namely $\hbar=2\pi \EXP^{\ii2\theta}$ with real $\theta\in]0,\pi/2[$. In this case, operators $\uop$ and $\vop$ cannot be neither self-adjoint nor unitary, but one can choose them to be normal and satisfying the relation
\begin{equation}
\uop\vop^\dagger=\vop^\dagger\uop,
\end{equation}
implying that operator~\eqref{Hams0} stays normal, and thus its spectral problem still makes sense.

In this paper, motivated by topological string applications, we address the spectral problem of operator~\eqref{Hams0} in the special case $m=1$ and $\hbar=2\pi\EXP^{\ii2\theta}$ with  $\theta\in[0,\pi/2[$. Our approach is based on solving Baxter's $T-Q$ equations in terms of analytic functions on the entire complex plane. The strong coupling regime $\theta\in ]0,\pi/2[$ is treated along the approach of \cite{MR2165903}, while the case $\theta=0$ (corresponding to $\hbar =2\pi$) is treated by solving the functional difference equation directly. We illustrate our analytical results with numerical calculations.

The organisation of the paper is as follows. In Section~\ref{sec1} we formulate the problem, fix the notation and conventions.
In Sections~\ref{sec2}--\ref{sec5}, we describe the solution of the problem in the strongly coupled regime. First, we introduce the main functional difference equation and describe some of its properties (Sec.~\ref{sec2}), then we construct the eigenfunctions with minimal number of poles (Sec.~\ref{sec3}), describe the spectral equations as conditions of cancelling the poles (Sec.~\ref{sec4}), and give some numerical analysis in the special case $\hbar=2\pi\ii$ (Sec.~\ref{sec5}). In the final Section~\ref{sec6}, we give a  solution of the case $\hbar=2\pi$.

\subsection*{Acknowledgements} We would like thank J\o rgen Ellegaard Andersen, Vladimir Bazhanov, Tobias Ekholm, Ludwig Faddeev, Vladimir Mangazeev, Marcos Mari\~no, Szabolc Zakany for valuable discussions. Our special thanks go to Alba Grassi, Marcos Mari\~no and Szabolc Zakany for sharing with us their results prior to publication. The work is partially supported by Australian Research Council, and the work of R.M.K. is partially supported by Swiss National Science Foundation.

\section{Background from basic quantum mechanics}\label{sec1}

Let  $\xop$ and $\pop$ be normalised  self-adjoint quantum mechanical position and momentum operators in the Hilbert space $L^2(\R)$ defined by the functional equalities
\begin{equation}\label{x-p-op}
\xop \psi(x)=x \psi(x),\quad \pop \psi(x)=\frac1{2\pi\ii}\frac{\partial \psi(x)}{\partial x}
\end{equation}
so that the Heisenberg commutation relation between them takes the form
\begin{equation}
[\pop,\xop]:=\pop\xop-\xop\pop=\frac{1}{2\pi\ii},
\end{equation}
i.e. with specific choice of Planck's constant  $h=2\pi \hbar=1$. As these are unbounded operators, it is assumed  that the function $\psi(x)$  in \eqref{x-p-op} is taken from the respective domains $D(\xop)$ and $D(\pop)$.

For each non-zero complex parameter $\bb$, there are  two pairs of exponential operators
\begin{equation}\label{expop1}
\uop:=\EXP^{2\pi\bb\xop}\;,\quad \vop:=\EXP^{2\pi\bb\pop}\;,
\end{equation}
and
\begin{equation}\label{expop2}
\bar{\uop}:=\EXP^{2\pi\bb^{-1}\xop}\;,\quad \bar{\vop}:=\EXP^{2\pi\bb^{-1}\pop}\;,
\end{equation}
which satisfy the Weyl commutations relations
\begin{equation}
\uop\vop=q^2\vop\uop\;,\quad
\bar{\vop}\bar{\uop}=\bar{q}^2\bar{\uop}\bar{\vop}\;,
\end{equation}
where
\begin{equation}\label{q-qb}
q:=\EXP^{\ii\pi\bb^2}\;,\quad \bar{q}:=\EXP^{-\ii\pi\bb^{-2}}\;,
\end{equation}
and all other commutation relations are trivial. In coordinate representation~\eqref{x-p-op}, we have
\begin{equation}\label{u-v-op}
\uop\psi(x)=\EXP^{2\pi\bb x}\psi(x),\quad \vop\psi(x)=\psi(x-\ii\bb),
\end{equation}
and 
\begin{equation}\label{bu-bv-op}
\bar\uop\psi(x)=\EXP^{2\pi\bb^{-1} x}\psi(x),\quad \bar\vop\psi(x)=\psi(x-\ii\bb^{-1}).
\end{equation}
The case where parameter $\bb$ is pure imaginary corresponds to unitary exponentials, while in this paper
we will mainly assume the \emph{strong coupling} regime
\begin{equation}\label{sc}
\bb=\EXP^{\ii\theta}\;,\quad \theta\in]0,\pi/2[\;\subset \R,
\end{equation}
which corresponds to unbounded exponentials. In this latter case, the operation  bar in \eqref{expop2}--\eqref{q-qb} and \eqref{bu-bv-op} is the complex conjugation for numbers and Hermitian conjugation for operators. We will keep that notation in a broader context of analytically continued parameters and variables always assuming that
\begin{equation}
|q|<1\quad\Leftrightarrow\quad \Im \bb^2>0\quad \Leftrightarrow\quad  |\bar{q}|<1\;.
\end{equation}
%Let us also denote
%\begin{equation}
%\eta\;:=\;\frac{1}{2}(\bb+\bb^{-1}).
%\end{equation}
As in \eqref{x-p-op}, the function $\psi(x)$ in \eqref{u-v-op} and \eqref{bu-bv-op} is assumed to belong to respective domains of unbounded operators. For example, the domain of operator $\vop$ consists of analytic  functions in the open strip between the real axis $\R\subset\C$ and the line  $-\ii\bb+\R$, which are continuous on the closure of the strip, and which restrict to square integrable functions on the boundary of the strip.

We remark that  the simultaneous appearance in a theory of both pairs of exponential operators \eqref{expop1} and \eqref{expop2} is a manifestation of \emph{Faddeev's modular duality}~\cite{MR1345554}.

%The coordinate representation of the Heisenberg--Weyl algebra is given by the formulae
%\begin{equation}
%\langle x|\xop=x\langle x|\;,\quad
%\langle x|\vop=\langle x-\ii\bb|\;,\quad
%\langle x|\bar{\vop}=\langle x-\ii\bb^{-1}|
%\end{equation}
%where we use Dirac's notation, a common practice in physics literature. In more explicit terms of a %\textbf{\emph{entire}} 
%wave function
%\begin{equation}
%\psi(x) = \langle x|\psi\rangle\;,
%\end{equation}
%this gives
%\begin{equation}
%\langle x|\xop|\psi\rangle=x\psi(x)\;,\quad
%\langle x|\vop|\psi\rangle=\psi(x-\ii\bb)\;,\quad
%\langle x|\bar{\vop}|\psi\rangle=\psi(x-\ii\bb^{-1})
%\end{equation}
%where we assume that $\psi(x)$ is in the respective domains of those unbounded operators so that the right hand sides make sense in the Hilbert space $L^2(\R)$ of square integrable functions on the real line. For example, the domain of operator $\vop$ is given by analytic  functions in the open strip between the real axis and the parallel line containing the point $-\ii\bb$, these functions are continuous on the closure of the strip, and they restrict to square integrable functions on the boundary of the strip.

\subsection*{Statement of the problem}

We address the spectral problem  of the following pair of commuting operators (to be called \emph{Hamiltonians}):
\begin{equation}\label{Hams}
\boldsymbol{H}=\vop+\vop^{-1} + \uop + \uop^{-1},\quad
\bar{\boldsymbol{H}}=\bar{\vop}+\bar{\vop}^{-1} + \bar{\uop}+\bar{\uop}^{-1}\;.
\end{equation}
 In the strong coupling regime~\eqref{sc}, $\bar{\boldsymbol{H}}=\boldsymbol{H}^\dagger$ so that commutativity of $\boldsymbol{H}$ and $\bar{\boldsymbol{H}}$ means that  $\boldsymbol{H}$ is a normal operator.
In the coordinate representation~\eqref{x-p-op}, the spectral problems
\begin{equation}
\boldsymbol{H}\psi = \energy \psi,\quad \bar{\boldsymbol{H}} \psi = \bar{\energy }\psi
\end{equation}
correspond to the following pair of functional difference equations:
\begin{align}\label{DEs}
 \psi(x+\ii\bb)+\psi(x-\ii\bb) &= (\energy -2\cosh(2\pi\bb x)) \psi(x),\\
\label{DEsc}\psi\left(x-\ii\bb^{-1}\right) + \psi\left(x+\ii\bb^{-1}\right) &= (\bar{\energy } - 2\cosh(2\pi\bb^{-1} x))\psi(x).
\end{align}
By definition, solving these spectral problems corresponds to construction of a function $\psi(x)$ and determination of  specific values of $\energy $ and $\bar{\energy }$ such that 
\begin{enumerate}\label{conditions}
\item $\psi(x)$ is analytic in an open strip around the real axis whose closure contains $\pm\ii \bb$ and  $\pm\ii \bb^{-1}$, and it is continuous on the closure of that strip;
\item $\psi(x)$ is square integrable on  the real axis and all of its translates by $\pm\ii \bb$ and  $\pm\ii \bb^{-1}$;
\item  equations~(\ref{DEs}), \eqref{DEsc} are satisfied for all $x\in\R$.
\end{enumerate}

\bigskip
\begin{remark}
 As we will see, the solution of the spectral problems~\eqref{DEs}, \eqref{DEsc} is given by analytic functions on the entire complex plane.
\end{remark}
\noindent\textbf{Conventions.} Throughout the text, each time when we use the variables $u$, $\bar{u}$ and $x$, we assume that they are related  as follows:
\begin{equation}
u=\EXP^{2\pi\bb x}\;,\quad \bar{u} = \EXP^{2\pi x/\bb}\;.
\end{equation}
Similarly,  as soon as we use the variables $s$, $\bar{s}$ and $\sigma$, we assume that they are related as follows:
\begin{equation}
s=\EXP^{2\pi\bb\sigma}\;,\quad \bar{s}=\EXP^{2\pi\sigma/\bb}\;.
\end{equation}
Else, if $\ds f = f(u,\energy ,q)$, then we write $\ds \bar{f}=\overline{f(u,\energy ,q)}$ meaning  the  analytic continuation of the complex conjugation in the strong coupling regime.

\begin{remark} In the limit $\bb\to 0$ one has
\begin{equation}
\boldsymbol{H} = 4+(2\pi\bb)^2(\pop^2+\xop^2)+\mathcal{O}(\bb^4)\;,
\end{equation}
what justifies the term ``oscillator''.
\end{remark}

\section{The main functional equation}\label{sec2}

In this section, we study the following functional finite difference equation 
\begin{equation}\label{chi-equation}
f\left(\frac{u}{q^2}\right) + q^2 u^2 f(q^2 u) = (1-\energy u+u^2) f(u).
\end{equation}
We can rewrite  it as a first order difference matrix equation
 \begin{equation}\label{mform}
\left(\begin{array}{cc}
f\left(\frac{u}{q^2}\right)\\
%[4mm]
f(u)
\end{array}\right) \;=\;
L(u)\left(\begin{array}{cc}
f\left(u\right)\\
%[4mm]
f(q^2u)
\end{array}\right)
\end{equation}
where
\begin{equation}\label{lop}
L(u) := \left(\begin{array}{cc}
\ds 1-\energy u+u^2 & \ds -q^2u^2\\
%[4mm]
\ds 1 & \ds 0
\end{array}\right).
\end{equation}
Upon iteration of \eqref{mform}, we obtain 
\begin{equation}\label{itmform}
\left(\begin{array}{cc}
f\left(\frac{u}{q^2}\right)\\
%[4mm]
f(u)
\end{array}\right) \;=\;
M_{n}(u)\left(\begin{array}{cc}
f\left(q^{2n-2}u\right)\\
%[4mm]
f(q^{2n}u)
\end{array}\right),\quad \forall n\in\Z_{>0},
\end{equation}
where
\begin{equation}
M_n(u):=L(u)L(q^2u)\cdots L(q^{2(n-1)}u)%,\quad n\in\Z_{\ge0},
\end{equation}
is a sequence of matrix valued polynomials in $u$ (the degree of $M_n(u)$ is $2n$) which satisfies the following obvious relations
\begin{equation}
M_{m+n}(u)=M_{m}(u)M_n(q^{2m}u),\quad \forall m,n\in\Z_{\ge0}.
\end{equation}
\subsection{The regular solution $\bchi_q(u,\energy )$}
The limit $n\to\infty$ leads to a matrix valued entire function  $M_\infty(u)$, i.e. a  holomorphic function on the entire complex plane of the variable $u$.  This follows from the following functional analytic argument. 

For any $r\in\R_{>0}$, let $V_r$ be the vector space  of  $\C^2$-valued continuous functions on the closed disk 
\begin{equation}
D_r:=\{z\in \C\ \vert\ |z|\le r\}.
\end{equation}
It is a Banach space with respect to the norm
\begin{equation}
\| g\|_r:=\max_{z\in D_r}\|g(z)\|,\quad \|g(z)\|:=\max_{i\in\{1,2\}}|g_{i}(z)|,\quad \forall g\in V_r.
\end{equation}
We have the following monotonicity property of the norms with different disk radii:
\begin{equation}
s\le r\quad\Rightarrow\quad \left\|g\vert_{D_s}\right\|_s\le\|g\|_r,\quad \forall g\in V_r.
\end{equation}

\begin{lemma}\label{lem1}
For any $r\in\R_{>0}$, the sequence of restrictions $M_n\vert_{D_r}$, $n\in\Z_{n>0}$, 
forms a (fundamental) Cauchy sequence with respect to the induced norm in the (Banach) algebra of bounded linear operators over $V_r$.
\end{lemma}
\begin{proof}
On the level of matrix valued continuous functions $A\colon D_r\to \operatorname{End}(\C^2)$, the induced norm is calculated as follows:
\begin{equation}
\| A\|_r=\max_{z\in D_r}\| A(z)\|,\quad \|A(z)\|:=\max_{i\in\{1,2\}}\sum_{j=1}^2|A_{i,j}(z)|.
\end{equation}
Using this formula and the explicit form~\eqref{lop} of the polynomial matrix $L(u)$, we have
\begin{equation}\label{normlop}
 \| L\|_r\ge 1,\quad \forall r\in\R_{\ge 0},
\end{equation}
and
\begin{equation}\label{normlop1}
\| L\|_r\le 1 +r(r+1)T,\quad T:=\max_{i\in\{1,2\}}\|L_i\|\ge 1+|q|^2,
\end{equation}
where $L_i\in\operatorname{End}(\C^2)$ are the expansion coefficients of the matrix $L(u)$:
\begin{equation}
L(u)=\sum_{i=0}^2 L_i u^i,
\end{equation}
the coefficient $L_0=L(0)$ being a projection matrix of unit norm:
\begin{equation}
L_0^2=L_0,\quad \| L_0\|=1.
\end{equation}
That allows us to get a uniform in $m$ upper bound  for the norms of $M_m(u)$:
\begin{equation}\label{est1}
\|M_m\|_r\le\max_{z\in D_r}\prod_{i=0}^{m-1}\|L(zq^{2i})\|\le\prod_{i=0}^{m-1}\| L\|_{r|q|^{2i}}\le \prod_{i=0}^{\infty}\| L\|_{r|q|^{2i}}=:\nu(r),
\end{equation}
where the infinite product converges due to inequality~\eqref{normlop1}. On the other hand, the norms of the coefficients in the expansions
\begin{equation}
M_m(u)=\sum_{k=0}^{2m} M_{m,k}u^k,\quad
M_{m,k}:=\sum_{\substack{0\le i_1,\dots, i_m\le 2\\ i_1+\dots+i_m=k}}L_{i_1}\dots L_{i_m} q^{2\sum_{j=1}^{m-1}ji_{j+1}}.
\end{equation}
are also bounded  from above uniformly on $m$:
\begin{equation}
 \| M_{m,k}\|\le \sum_{\substack{0\le i_1,\dots, i_m\le 2\\ i_1+\dots+i_m=k}}\|L_{i_1}\|\dots \| L_{i_m}\|  |q|^{2\sum_{j=1}^{m-1}ji_{j+1}}
 \le  \frac{3T^k}{(|q|^2;|q|^2)_{m-1}}\le \frac{3T^k}{(|q|^2;|q|^2)_\infty}
\end{equation}
with the standard notation for the deformed Pochhammer symbols
\begin{equation}
(x;q)_n:=\prod_{i=0}^{n-1}(1-xq^i).
\end{equation}
For any $m,n\in\Z_{>0}$, there exists a matrix valued polynomial $A_{m,n}(u)$ of degree $2\max(m,n)-1$ such that
\begin{equation}
M_{m}(u)-M_n(u)=u A_{m,n}(u).
\end{equation}
Explicitly, we have
\begin{equation}
A_{m,n}(u)=\sum_{k=0}^{2\max(m,n)-1}(M_{m,k}-M_{n,k}) u^{k} 
\end{equation}
so that
\begin{equation}\label{est2}
\|A_{m,n}\|_r\le\sum_{k=0}^{2\max(m,n)-1}(\|M_{m,k}\|+\| M_{n,k}\|) r^{k} \le  \frac{6\left(1-(rT)^{2\max(m,n)}\right)}{(|q|^2;|q|^2)_\infty(1-rT)}.
\end{equation}
Let $N_r\in\Z_{>0}$ be such that $rT|q|^{2N_r}<1$. Then for any $m\ge N_r$ and $n\ge0$, taking into account the formula
\begin{equation}
M_{m+n+1}(u)-M_{m+1}(u)=M_m(u)uq^{2m}A_{n+1,1}(uq^{2m})
\end{equation}
and the inequalities~\eqref{est1} and \eqref{est2}, we obtain
\begin{multline}
\|M_{m+n+1}-M_{m+1}\|_r\le  |q|^{2m}r\|M_m\|_r\|A_{n+1,1}\|_{r|q|^{2m}}\\
\le |q|^{2m}r\nu(r) \frac{6(1-(rT|q|^{2m})^{2n+2})}{(|q|^2;|q|^2)_\infty(1-rT|q|^{2m})}
\le |q|^{2m} \frac{6r\nu(r)}{(|q|^2;|q|^2)_\infty(1-rT|q|^{2N_r})}
\end{multline}
The obtained equality implies that $(M_m\vert_{D_r})_{m>0}$ is a Cauchy sequence in the Banach algebra of bounded linear operators over $V_r$.
\end{proof}
Lemma~\ref{lem1} implies the uniform convergence of the matrix coefficients of $M_n(u)$ on all compact subsets of $\C$. Indeed , let $f_n(u)$ be a matrix element of $M_n(u)$ and  $K\subset \C$  a compact set. Then there exists $r\in\R_{>0}$ such that $K\subset D_r$ and  thus
\begin{equation}
|f_{m}(u)-f_{n}(u)|\le \| M_m(u)-M_n(u)\|\le \| M_m-M_n\|_r,\quad \forall u\in K.
\end{equation}
Thus, $M_\infty(u)$ is an entire function on $\C$ as a uniformly convergent limit on all compact subsets of $\C$ of a sequence of polynomial functions. 

By taking into account the equalities
\begin{equation}
M_\infty(u)=L(u)M_\infty(q^2u)=M_\infty(u)L(0),\quad M_\infty(0)=L(0),
\end{equation}
we immediately arrive to the following structure of the matrix $M_\infty(u)$:
\begin{equation}\label{infpr}
M_\infty(u)=\left(
\begin{array}{cc}
\ds \bchi_q\left(q^{-2}u\right) & \ds 0 \\
%[4mm]
\ds \bchi_q(u) & \ds 0
\end{array}\right),
\end{equation}
where $ \bchi_q(u)= \bchi_q(u,\energy )$ is an entire function normalised so that  $ \bchi_q(0)=1$ and which solves the functional difference equation~\eqref{chi-equation}. 
On the other hand,  if write
 \begin{equation}\label{exmn}
M_n(u) = \left(\begin{array}{cc}
\ds a_n(u) & \ds b_n(u)\\
%[4mm]
\ds c_n(u) & \ds d_n(u)
\end{array}\right),
\end{equation}
then the matrix equality
\begin{equation}
L(u)M_n(q^2u)=M_n(u)L(q^{2n}u)
\end{equation}
implies that
\begin{equation}\label{bndn}
b_n(q^2u)=-q^{4n+2}u^2 c_n(u),\quad d_n(q^2u)=q^{4n}(a_n(u)-(1-\energy u+u^2)c_n(u)),
\end{equation}
and we explicitly see the rate at which $b_n(u)$ and $d_n(u)$ converge to zero functions. Furthermore, by taking into account the limits
\begin{equation}
\lim_{n\to\infty}a_n(u)=\bchi_q(q^{-2}u),\quad \lim_{n\to\infty}c_n(u)=\bchi_q(u),
\end{equation}
we obtain from \eqref{bndn}
\begin{equation}
\lim_{n\to\infty}q^{-4n}b_n(u)=-q^{-2}u^2\bchi_q(q^{-2}u),\quad \lim_{n\to\infty}q^{-4n}d_n(u)=-q^{-2}u^2\bchi_q(u),
\end{equation}
where in the last limit we have used the equation~\eqref{chi-equation} for  $\bchi_q(u)$.
\begin{theorem}[Uniqueness of the regular solution]\label{th1}
In the case $|q|<1$, let $f(u)$ be a solution of the functional equation \eqref{chi-equation} 
%admits a unique solution denoted as $\bchi_q(u)$ (or $\bchi_q(u,\energy )$ if it is needed to indicate explicitly the dependence on the variable $\energy $) 
which is regular (i.e. holomorphic) at $u=0$  and normalised so that $f(0)=1$. Then $f(u)=\bchi_q(u)$.
\end{theorem}
\begin{proof}
 Taking the limit $n\to \infty$ in the matrix equality~\eqref{itmform} and using the regularity and normalisation properties of $f(u)$ at $u=0$, we obtain the equality

\begin{equation}
\left(\begin{array}{cc}
f\left(q^{-2}u\right)\\
%[4mm]
f(u)
\end{array}\right) =
M_\infty(u)\left(\begin{array}{cc}
1\\
%[4mm]
1
\end{array}\right)=\left(\begin{array}{cc}
\bchi_q\left(q^{-2}u\right)\\
%[4mm]
\bchi_q(u)
\end{array}\right).
\end{equation}
\end{proof}
\begin{lemma}\label{lem2}
In the space of solutions of \eqref{chi-equation}, there is an involution which associates to any solution $f(u)$ another solution  $\check f(u)$ defined by
 \begin{equation}
\check f(u)=u^{-1} f\left(u^{-1}\right).
\end{equation}
\end{lemma}
\begin{proof}
 This is an easy direct verification.
\end{proof}

%Function $\bchi_q(u,\energy )$ is then defined by the following ordered infinite matrix product:
%\begin{equation}
%\left(\begin{array}{cc}
%\ds \bchi_q(\frac{u}{q^2},\energy ) & \ds 0 \\
%[4mm]
%\ds \bchi_q(u,\energy ) & \ds 0
%\end{array}\right) \;=\;
%L(u,\energy ) L(q^2u,\energy ) \cdots L(q^{2n}u,\energy ) \cdots \;=\;\prod_{n=0}^\infty L(q^{2n}u,\energy )\;.
%\end{equation}
%Since the product converges absolutely, $\bchi_q(u,\energy )$ is holomorphic function of $u$ with $\bchi_q(0,\energy )=1$. 

%The statement of Theorem \ref{th1} also follows from (\ref{chi-series},\ref{chin-recursion}).

%\begin{remark}
%In the limit $u\to\infty$, the function $\bchi_q(u,\energy )$ satisfies the following approximate functional equation:
%\begin{equation}\label{chi-lim}
%\bchi_q(\frac{u}{q^2},\energy )\;\simeq\; u^2 \bchi_q(u,\energy )\;.
%\end{equation}
%\end{remark}
Applying Lemma~\ref{lem2} to our regular solution $\bchi_q(u,\energy )$, we obtain a second solution of the recursion equation (\ref{chi-equation})
\begin{equation}
\check\bchi_q(u,\energy ):=u^{-1}\bchi_q(u^{-1},\energy )
\end{equation}
which is holomorphic in $\C_{\ne0}$. 
\subsection{Orthogonal polynomials associated to $\bchi_q(u,\energy )$}
The function $\bchi_q(u,\energy )$ can be expanded as a  power series
\begin{equation}
\bchi_q(u,\energy )=\sum_{n\ge0} c_n u^n,\quad c_0=1,
\end{equation}
with infinite radius of convergence. The coefficients here are functions of $\energy $ and $q$.
The functional equation~\eqref{chi-equation} in this case is translated into a system of recurrence equations on the coefficients
\begin{equation}
(1-q^{-2-2n})c_{n+1}=\energy c_n-(1-q^{2n})c_{n-1}, \quad n\in\Z_{\ge0},
\end{equation}
which, upon multiplication by $(q^{-2};q^{-2})_n$, is rewritten as a system of recurrence relations defining orthogonal polynomials $\chi_{q,n}(\energy )\in\C[\energy ]$, $n\in\Z_{n\ge0}$:
\begin{equation}\label{chin-recursion}
\chi_{q,0}(\energy )=1\;,\quad \chi_{q,n+1}(\energy ) = \energy  \chi_{q,n}(\energy ) + (q^n-q^{-n
})^2 \chi_{q,n-1}(\energy ),
\end{equation}
where
\begin{equation}
\chi_{q,n}(\energy ):=(q^{-2};q^{-2})_nc_n.
\end{equation}
Notice that the polynomials $\chi_{q,n}(\energy )$ also depend on the deformation parameter $q$ in symmetric way in the sense that they are unchanged under the replacement $q\mapsto q^{-1}$.
Here are the explicit forms of the first four polynomials
\begin{multline}
\chi_{q,0}(\energy )=1,\quad \chi_{q,1}(\energy )=\energy ,\quad \chi_{q,2}(\energy )=\energy ^2+(q-q^{-1})^2,\\
 \chi_{q,3}(\energy )=\energy (\energy ^2+(q^2-q^{-2})^2+(q-q^{-1})^2),\quad\ldots
\end{multline}

Thus, our solution $\bchi_q(u,\energy )$  is represented in the form of a power series
\begin{equation}\label{chi-series}
\bchi_q(u,\energy )\;:=\;\sum_{n\geq 0} \frac{\chi_{q,n}(\energy )}{(q^{-2};q^{-2})_n} u^n=\sum_{n\geq 0}(-1)^n q^{n(n+1)} \frac{\chi_{q,n}(\energy )}{(q^{2};q^{2})_n} u^n\;.
\end{equation}
In order to see explicitly how this series absolutely  convergences on the entire complex plane of the variable $u$, it suffices to find the growth rate of the polynomials $\chi_{q,n}(\energy )$ at large index $n$. To this end, we observe that the recurrence relation~\eqref{chin-recursion}
at large $n$ approaches  the form
\begin{equation}
\chi_{n+1} \simeq \energy  \chi_n + q^{-2n} \chi_{n-1}
\end{equation}
where we have suppressed the arguments $q$ and  $\energy $ of $\chi_{q,n}(\energy )$ and have taken into account the inequality $|q|<1$.
This equivalently can be rewritten as follows:
\begin{equation}\label{asb}
 \frac{\chi_{n+1}}{\chi_{n}}\frac{\chi_{n}}{\chi_{n-1}}\simeq \energy \frac{\chi_{n}}{\chi_{n-1}}+q^{-2n}\\
\quad \Leftrightarrow\quad\frac{\chi_{n}}{\chi_{n-1}}\simeq q^{-n}\quad \Leftrightarrow\quad \chi_{n}\simeq q^{-n^2/2}.
\end{equation}
The latter relation immediately implies the absolute convergence of the sum in \eqref{chi-series}.
\begin{theorem}
 The orthogonal polynomials $\chi_{q,n}(\energy )$ satisfy the following multiplication rules
 \begin{equation}
\chi_{q,m}(\energy )\chi_{q,n}(\energy )=\sum_{k=0}^{\min(m,n)} \frac{(q^{2m};q^{-2})_k(q^{2n};q^{-2})_k(q^{2(k-m-n)};q^{2})_k}{(q^2;q^2)_k}\chi_{q,m+n-2k}(\energy )
\end{equation}
\end{theorem}
\begin{proof}
 This is a direct check by recurrence  on $\max(m,n)$.
\end{proof}

\subsection{Non-linear first order functional difference equations}
The vector space of solutions of  \eqref{chi-equation} is a module over the algebra of $q^2$-periodic functions, i.e. functions $g(u)$ that satisfy the functional difference equation $g(q^2u)=g(u)$.
For any nontrivial solution $f(u)$ of \eqref{chi-equation}, the combinations
\begin{equation}
R_f(u):=\frac{f(q^{-2} u)}{f(u)}
\end{equation}
and 
\begin{equation}
P_f(u):=\frac{q^2u^2f(q^{2} u)}{f(u)}
\end{equation}
are invariant with respect to multiplication by  $q^2$-periodic functions, so that they can be used for characterisation of the equivalence classes of solutions of \eqref{chi-equation} differing by $q^2$-periodic functions. The two combinations are related to each other  through the functional equalities
\begin{equation}
P_f(u)R_f(q^2u)=q^2u^2,\quad P_f(u)+R_f(u)=1-\energy u+u^2.
\end{equation}

The second order functional  linear difference equation~\eqref{chi-equation} implies  first order non-linear functional difference equations for $R(u)=R_f(u)$
\begin{equation}\label{R-equation}
R(u)R\left(q^2u\right)=(1-\energy u+u^2)R\left(q^2u\right)-q^2u^2
\end{equation}
and for $P(u)=P_f(u)$
\begin{equation}\label{P-equation}
P(u)P\left(q^{-2}u\right)=(1-\energy u+u^2)P\left(q^{-2}u\right)-q^{-2}u^2.
\end{equation}
Notice that two equations are related to each other by the symmetry $q\mapsto q^{-1}$.

By using  \eqref{itmform} and \eqref{exmn}  %and \eqref{bndn}
 we easily obtain the following linear fractional transformation formulae
 \begin{equation}\label{fltf}
R(q^{2n} u)=\frac{d_n(u)R(u)-b_n(u)}{a_n(u)-c_n(u)R(u)},\quad \forall n\in\Z_{\ge0}.
\end{equation}
%and 
% \begin{equation}\label{fltf2}
%P(q^{2n} u)=\frac{q^2u^2(a_n(q^2u)-c_n(q^2u)R(q^2u))}{(a_n(u)-(1-\energy u+u^2)c_n(u))R(q^2u)+q^2u^2c_n(u)},\quad \forall n\in\Z_{\ge0}.
%\end{equation}
\begin{theorem}\label{thm3}
Let  $ R,P\colon U\to \C$, $U\subset \C$,  be solutions of \eqref{R-equation}  and \eqref{P-equation}, respectively, and $z\in U$ be such that $\bchi_q(z)\ne0$.
Then the sequences $(R(q^{2n}z))_{n\ge0}$ and $(P(q^{2n}z))_{n\ge0}$ converge with the limits
\begin{equation}\label{lim1}
\lim_{n\to\infty}R(q^{2n}z)=\left\{
\begin{array}{cc}
  0&\text{if}\ R(z)\ne R_{\bchi_q}(z) ; \\
  1&\text{otherwise.}
\end{array}
\right.
\end{equation}
and 
\begin{equation}\label{lim2}
\lim_{n\to\infty}P(q^{2n}z)=\left\{
\begin{array}{cc}
  1&\text{if}\ P(z)\ne P_{\bchi_q}(z) ; \\
  0&\text{otherwise.}
\end{array}
\right.
\end{equation}
\end{theorem}
\begin{proof} Formula~\eqref{lim2} follows from \eqref{lim1}  if we define  $R(u):=1-\energy u+u^2-P(u)$.

If $R(z)=R_{\bchi_q}(z)$, then formula~\eqref{fltf} with $u=z$ implies that $R(q^{2n}z)=R_{\bchi_q}(q^{2n}z)$
% and $P(q^{2n}z)=P_{\bchi_q}(q^{2n}z)$ 
for all $n\in\Z_{\ge0}$. 
Thus,
\begin{equation}
\lim_{n\to\infty}R(q^{2n}z)=\lim_{n\to\infty}R_{\bchi_q}(q^{2n}z)=R_{\bchi_q}(0)=1
\end{equation}
%and 
%\begin{equation}
%\lim_{n\to\infty}P(q^{2n}z)=\lim_{n\to\infty}P_{\bchi_q}(q^{2n}z)=\lim_{n\to\infty}\frac{q^{4n+2}u^2}{R_{\bchi_q}(q^{2n+2}z)}=0
%\end{equation}
by the regularity and the normalisation properties of $\bchi_q(u)$ at $u=0$.

If we assume that $R(z)\ne R_{\bchi_q}(z)$, then in equality~\eqref{fltf} with $u=z$ the denominator is distinct from zero for sufficiently large $n$ since it converges to a nonzero number:
\begin{equation}
\lim_{n\to\infty}(a_n(z)-c_n(z)R(z))=a_\infty(z)-c_\infty(z)R(z)
=\bchi_q(q^{-2}z)-\bchi_q(z)R(z)=\bchi_q(z)(R_{\bchi_q}(z)-R(z)),
\end{equation}
while the numerator converges to $b_\infty(z)R(z)-d_\infty(z)=0$. 
%In the case of $P(q^{2n} z)$, equality~\eqref{fltf2} gives
 %\begin{multline}
%\lim_{n\to\infty}P(q^{2n}z)=\frac{q^2z^2(\bchi_q(z)-\bchi_q(q^2z)R(q^2z))}{(\bchi_q(q^{-2}z)-(1-\energy z+z^2)\bchi_q(z))R(q^2z)+q^2z^2\bchi_q(z)}\\=\frac{q^2z^2(\bchi_q(z)-\bchi_q(q^2z)R(q^2z))}{q^2z^2\bchi_q(z)-q^2z^2\bchi_q(q^{2}z)R(q^2z)}
%=\frac{\bchi_q(z)-\bchi_q(q^2z)R(q^2z)}{\bchi_q(z)-\bchi_q(q^{2}z)R(q^2z)}\\
%=\frac{R_{\bchi_q}(q^2z)-R(q^2z)}{R_{\bchi_q}(q^2z)-R(q^2z)}=1.
%\end{multline}
\end{proof}
\subsection{Wronskian of two solutions of the main functional equation}
An important function associated  with any two solutions of the functional equation~\eqref{chi-equation} is their Wronskian or, more precisely, the functional difference analogue of the Wronskian in the theory of linear second order ordinary differentian equations. As we will see below, any Wronskian satisfies a particularly simple difference equation of the form
\begin{equation}\label{difeqwron}
W(q^2u)\;=\;\frac{1}{q^2u^2} W(u)\quad \Leftrightarrow\quad W\left(u/q\right)=u^2W(qu).
\end{equation}

\begin{definition-proposition}
 Let $f(u)$ and $g(u)$ satisfy equation~\eqref{chi-equation}. Then their \emph{Wronskian} $[f,g](u)$ defined by
 \begin{equation}
[f,g](u):=f\left(q^{-2}u\right)g(u)-g\left(q^{-2}u\right)f(u)
\end{equation}
satisfies equation~\eqref{difeqwron}.
\end{definition-proposition}
\begin{proof} Substructing \eqref{chi-equation} for $g(u)$ multiplied by $f(u)$ from \eqref{chi-equation} for $f(u)$ multiplied by $g(u)$, we obtain 
 \begin{multline}
\left(f\left(q^{-2}u\right) + q^2 u^2 f(q^2 u) \right)g(u)=\left(g\left(q^{-2}u\right) + q^2 u^2 g(q^2 u) \right)f(u)\\
\Leftrightarrow
f\left(q^{-2}u\right)g(u) -g\left(q^{-2}u\right)f(u)= q^2 u^2 \left(g(q^2 u)f(u)-f(q^2 u)g(u)\right)\\
\Leftrightarrow
[f,g](u)= q^2 u^2[f,g](q^2u).
\end{multline}
\end{proof}
Of particular importance for the sequel will be the Wronskian of $\bchi$ and $\check\bchi$:
\begin{equation}\label{wron}
[\bchi,\check\bchi](u)\;=\;\bchi\left(q^{-2}u\right)\check\bchi(u)-\check\bchi\left(q^{-2}u\right)\bchi(u)\;.
\end{equation}
Here and hereafter we shorten the notation $\bchi(u)=\bchi_q(u,\energy )$ and so on. 

\begin{lemma}
 The Wronskian $[\bchi,\check\bchi]$ is not the identically zero function.
\end{lemma} 
\begin{proof}
 Using the power series expansions, we obtain a Laurent series
 \begin{equation}
[\bchi,\check\bchi](u)=\sum_{n\in\Z}[\bchi,\check\bchi]_n u^n,
\end{equation}
where the residue $[\bchi,\check\bchi]_{-1}$ is given explicitly by the absolutely convergent series
\begin{equation}
[\bchi,\check\bchi]_{-1}=\sum_{m\ge0}\left(\frac{\chi_{q,m}(\energy )}{(q^{-2};q^{-2})_m}\right)^2(q^{-2m}-q^{2m+2})
\end{equation}
where each term is non-negative if the variables $\energy $ and $q$ are real with $0<q<1$. Thus, dropping all but the very first term,
 we obtain 
\begin{equation}
[\bchi,\check\bchi]_{-1}\ge 1-q^2>0.
\end{equation}
\end{proof}
\subsection{Parameterization in terms of $\theta$-functions}
Consider Jacobi's $\theta$-function
\begin{equation}
\theta_1(u,q)=\frac{1}{\ii} \sum_{n\in\mathbb{Z}} (-1)^n q^{(n+1/2)^2} u^{n+1/2}\;.
\end{equation}
It has the properties
\begin{equation}\label{thetaprop}
\theta_1(1,q)=0,\quad 
\theta_1(u^{-1},q)= -\theta_1(u,q),\quad \theta_1(q^2u,q)=-\frac{1}{qu}\theta_1(u,q),
\end{equation}
and
\begin{equation}\label{thetamodprop}
\overline{\theta_1(u,q)}=\bb \EXP^{\pi\ii/4-\ii\pi x^2} \theta_1(u,q).
\end{equation}

Any solution $W(u)$ of  \eqref{difeqwron}  in the form of a Laurent series in $u$, can be brought to the form
\begin{equation}
W(u) =\varrho \theta_1(su,q)\theta_1(s^{-1}u,q).
\end{equation}
for some constants $s$ and  $\varrho 
$. The properties of the $\theta$-function given by \eqref{thetaprop} and \eqref{thetamodprop} imply that
\begin{equation}\label{modularW}
\overline{W(u)}=\ii\bb^2 \bar\varrho\varrho^{-1}\EXP^{-2\pi\ii(\sigma^2+x^2)}W(u),\quad u=\EXP^{2\pi\bb x},\quad s=\EXP^{2\pi\bb \sigma}.
\end{equation}

In the case when $W=[\bchi,\check\bchi]$, the  variables $s$ and $\varrho$ will be determined in terms of $\energy $ and $q$:
\begin{equation}
s=s(\energy ,q),\quad \varrho=\varrho(\energy ,q).
\end{equation}
The first of these dependences is of particular importance for us, and it will be the subject of a numerical study.
\subsection{The main functional equation with $q$ replaced by $q^{-1}$}
We turn now to the discussion of the functional equation~(\ref{chi-equation}) with $q$ replaced by $q^{-1}$:
\begin{equation}\label{chi-equation-2}
f(q^2u) + \frac{u^2}{q^2} f\left(\frac{u}{q^2}\right) = (1-\energy u+u^2) f(u)\;.
\end{equation}
Let $\bchi_{q^{-1}}(u,\energy )$ be the normalised solution \eqref{chi-series} with $q$ replaced by $q^{-1}$, i.e. a power series of the form
\begin{equation}\label{anti-chi-series}
\bchi_{q^{-1}}(u,\energy )\;\simeq\;\sum_{n\geq 0} \frac{\chi_{q,n}(\energy )}{(q^{2};q^{2})_n} u^n
\end{equation}
where we use the semi-equality symbol $\simeq$ to indicate the fact that the series  is now only a formal power series as it has zero radius of convergence. This is an immediate consequence of the asymptotic behaviour of $\chi_{q,n}(\energy )$ at large $n$ given by \eqref{asb}. Thus, there does not exist a regular at $u=0$ solution of equation~\eqref{chi-equation-2}. Nonetheless, the formal power series solution~\eqref{anti-chi-series} is interpreted as the asymptotic expansion of a true solution which is not analytic  at $u=0$. Before formulating this result we do some preparatory work.
\begin{lemma}\label{l3}
 Let $f(u)$ and $W(u)$ satisfy the functional equations~\eqref{chi-equation} and \eqref{difeqwron} respectively. Then the ratio
$
f(u)/W(u)
$
satisfies \eqref{chi-equation-2}.
\end{lemma}
\begin{proof}
 This is a direct verification.
\end{proof}
\begin{lemma}
 Let $f(u)$ and $g(u)$ satisfy the functional equations~\eqref{chi-equation} and \eqref{chi-equation-2} respectively. Then the combination
 \begin{equation}
\langle f,g\rangle (u):=f\left(\frac{u}{q^2}\right)g(u)  - \frac{u^2}{q^2} g\left(\frac{u}{q^2}\right)f(u)
\end{equation}
is $q^2$-periodic, i.e.
\begin{equation}\label{constrel}
\langle f,g\rangle (u)=\langle f,g\rangle (q^{2}u),\quad \forall u.
\end{equation}
\end{lemma}
\begin{proof} Substructing \eqref{chi-equation-2} multiplied by $f(u)$ from \eqref{chi-equation} multiplied by $g(u)$ , we obtain 
 \begin{multline}
\left(f\left(\frac{u}{q^2}\right) + q^2 u^2 f(q^2 u) \right)g(u)=\left(g(q^2u) + \frac{u^2}{q^2} g\left(\frac{u}{q^2}\right) \right)f(u)\\
\Leftrightarrow
f\left(\frac{u}{q^2}\right)g(u)  - \frac{u^2}{q^2} g\left(\frac{u}{q^2}\right)f(u)=f(u)g(q^2u)- q^2 u^2 g(u)f(q^2 u) 
\end{multline}
where the latter equality is exactly \eqref{constrel}.
\end{proof}
\begin{theorem}
The formal power series~\eqref{anti-chi-series} is the asymptotic expansion  of the  following solution of the functional equation~\eqref{chi-equation-2}
\begin{equation}\label{anti-chi}
\bchi_{q^{-1}}(u,\energy ):=\frac{\check\bchi_q(u,\energy )}{[\bchi_q,\check\bchi_q](u)}\;
\end{equation}
where  the Wronskian $[\bchi_q,\check\bchi_q](u)$ is defined in \eqref{wron}.
\end{theorem}

\begin{proof} The fact that $\bchi_{q^{-1}}(u,\energy )$ defined in \eqref{anti-chi} satisfies \eqref{chi-equation-2} follows from Lemma~\ref{l3}. The rest of the proof is based on the following two easily verifiable functional identities
\begin{equation}\label{crochet1}
\langle\bchi_q,\bchi_{q^{-1}}\rangle(u)=1\quad\Leftrightarrow\quad\bchi_q\left(u\right)\bchi_{q^{-1}}(q^2u) - q^2u^2\bchi_{q^{-1}}\left(u\right)\bchi_q(q^2u)\;=\;1
\end{equation}
and 
\begin{equation}\label{crochet2}
\langle\check\bchi_q,\bchi_{q^{-1}}\rangle(u)=0\quad\Leftrightarrow\quad P_{\check\bchi_q}(u)=
\frac{\bchi_{q^{-1}}(q^{2}u)}{\bchi_{q^{-1}}(u)}.
\end{equation}
The Wronskian $[\bchi,\check\bchi](u)$ can equivalently  be rewritten in the form
\begin{equation}
%R_{\check\bchi_q}(u)=q^{-2}u^2\bchi_q(u)\bchi_{q^{-1}}(q^{-2}u)(R_{\bchi_q}(u)-R_{\check\bchi_q}(u))\\
%\quad\Leftrightarrow\quad
[\bchi,\check\bchi](u)=\bchi(u)\check\bchi(u)(R_{\bchi}(u)-R_{\check\bchi}(u))
\end{equation}
which implies that if $[\bchi,\check\bchi](z)\ne0$ then we have $R_{\bchi}(z)\ne R_{\check\bchi}(z)$, and by Theorem~\ref{thm3} combined with equality~\eqref{crochet2}  we conclude that
\begin{equation}
\lim_{n\to\infty }\frac{\bchi_{q^{-1}}(q^{2+2n}z)}{\bchi_{q^{-1}}(q^{2n}z)}=\lim_{n\to\infty }P_{\check\bchi}(q^{2n}z)=1.
\end{equation}
On the other hand, from equality~\eqref{crochet1} we also obtain a sequence of equalities
\begin{equation}
\frac{1}{\bchi_{q^{-1}}\left(q^{2n}z\right)}=
\frac{\bchi_q\left(q^{2n}z\right)\bchi_{q^{-1}}(q^{2+2n}z)}{\bchi_{q^{-1}}\left(q^{2n}z\right)} - q^{2+4n}z^2\bchi_q(q^{2+2n}z),\quad \forall n\in\Z_{\ge0}
\end{equation}
so that
\begin{equation}
\lim_{n\to\infty }\frac{1}{\bchi_{q^{-1}}\left(q^{2n}z\right)}=1\quad\Leftrightarrow\quad\lim_{n\to\infty }\bchi_{q^{-1}}\left(q^{2n}z\right)=1.
\end{equation}
Now, expanding in a formal power series $\bchi_{q^{-1}}\left(u\right)$ around $u=0$, we recover the series~\eqref{anti-chi-series} as it is fixed uniquely by the functional difference equation~\eqref{chi-equation-2} and the initial condition
$\bchi_{q^{-1}}\left(0\right)=1$.
\end{proof}
% \hfill $\square$

\section{Behavior at infinity and Ansatz for the eigenfunction} \label{sec3}

In the limit $x\to -\infty$, equation (\ref{DEs}) is approximated by the equation
\begin{equation}
 \psi(x+\ii\bb)+\psi(x-\ii\bb) = -\EXP^{-2\pi\bb x}\psi(x),
\end{equation}
where, in the left hand side, any one of the two terms  can be dominating giving rise to two equations
\begin{equation}
\psi(x+\epsilon\ii\bb)=-\EXP^{-2\pi\bb x} \psi(x),\quad \epsilon\in\{\pm1\},
\end{equation}
with particular  solutions
\begin{equation}
\psi(x)=\EXP^{\epsilon\ii\pi x^2 + 2\pi\eta x}
\end{equation}
which describe the asymptotic behavior at $x\to -\infty$ of two solutions  of \eqref{DEs} and \eqref{DEsc}
\begin{equation}
\psi_\epsilon(x)= \EXP^{\epsilon\ii\pi x^2 + 2\pi\eta x}\phi_\epsilon(x),\quad \left.\phi_\epsilon (x)\right\vert_{x\to-\infty}=\mathcal{O}(1),\quad \epsilon\in\{\pm1\}.
\end{equation}
Inspection of the equations for $\phi_{\epsilon}(x)$ gives solutions
\begin{equation}
\phi_\epsilon(x)=\bchi_{q^{-\epsilon}}(u)\overline{\bchi_{q^{\epsilon}}(u)},\quad  \epsilon\in\{\pm1\},
\end{equation}
up to multiplication by doubly periodic functions of $x$.
Thus, a general solution for \eqref{DEs} and \eqref{DEsc} exponentially decays at $x\to -\infty$ and is given by the formula
\begin{equation}\label{ansatz-1}
\psi(x) =\EXP^{2\pi\eta x}\sum_{\epsilon\in\{\pm1\}}A_\epsilon\EXP^{\epsilon \ii\pi x^2}\phi_\epsilon(x)
\end{equation}
where, in general, $A_\epsilon$ are doubly periodic functions of $x$. 
As non-constant doubly periodic functions provide unwanted extra poles of the eigenfunction, we assume $A_\epsilon $ to be constants. Moreover, using the modular transformation formula~\eqref{modularW}
 with $W(u)=[\bchi,\check\bchi](u)$, and choosing
 \begin{equation}
A_{1}=\bb^{-1}\varrho   \EXP^{\pi\ii\sigma^2-\xi\pi\ii/4},\quad A_{-1}=\bar{A_1}=\bb\bar\varrho \EXP^{\xi\pi\ii/4-\pi\ii\sigma^2},
\end{equation}
where $\xi^2=1$, we arrive to the final form for the eigenfunction:
\begin{equation}\label{final-ansatz}
\psi(x) = \bb^{-1} \EXP^{\pi\ii\sigma^2-\xi\pi\ii/4}\EXP^{2\pi\eta x + \ii\pi x^2}\frac{ \check\bchi(u)\overline{\bchi(u)}+ \xi \bchi(u)\overline{\check\bchi(u)} }{\theta_1(su,q)\theta_1(s^{-1}u,q)}.
\end{equation}
One can easily verify that
\begin{equation}
\psi(-x)=\xi\psi(x),
\end{equation}
so that the parameter $\xi$ is identified with the parity of the state. We also have the  exponential decay at both infinities
\begin{equation}
|\psi(x)|\sim \EXP^{-2\pi\eta |x|},\quad x\to\pm\infty,
\end{equation}
and the reality property
\begin{equation}
\overline{\psi(x)}=\psi(x).
\end{equation}

\section{The quantization condition}\label{sec4}
In the previous section, we have constructed a solution~\eqref{final-ansatz} for \eqref{DEs} and  \eqref{DEsc} which exponentially decays at plus and minus infinities and has a minimal number of poles. The quantization condition is the analyticity condition for (\ref{final-ansatz}) in the strip 
\begin{equation}
S_\bb:=\left\{z\in\C\ \vert\ |\Im z|<\max(\Re \bb,\Re\bb^{-1})\right\}.
\end{equation}
As a matter of fact, this happens to be equivalent to the absence of poles in the entire complex plane $\mathbb{C}$, at least in the case when all zeros of the wronskian $[\bchi,\check\bchi](u)$ are simple.

Let  us define a function 
\begin{equation}\label{G}
G_q(u,\energy ):=\frac{\bchi_q(u,\energy )}{\check\bchi_q(u,\energy )},
\end{equation}
which has the (anti)symmetry property
\begin{equation}\label{antisym}
G_q(u,\energy )G_q(1/u,\energy )=1,\quad \forall u\in\C.
\end{equation}

\begin{theorem}\label{spectral-theorem} Let $\energy=\energy(\sigma)$ be such that 
\begin{equation}
W(u):=[\bchi_q,\check\bchi_q](u)=\varrho\theta_1(su)\theta_1(s^{-1}u),\quad \forall u\in\C,
\end{equation}
and assume that $s\not\in \pm q^{\Z}$ (recall that $s=s(\sigma)=\EXP^{2\pi\bb\sigma}$). Then the eigenfunction $\psi(x)$  defined by (\ref{final-ansatz}) does not have poles in the strip $S_\bb$ if the variable $\sigma$ is such that
\begin{equation}\label{spectral}
G_q(s,\energy ) = -\xi \overline{G_q(s,\energy )}.
\end{equation}
%This equation is the quantisation condition for our spectral problem. 
In that case, the eigenfunction $\psi(x)$ is an entire function on the whole complex plane $\C$. 
\end{theorem}  
\begin{proof}
 The zeros of the denominator in (\ref{final-ansatz}) are given by
\begin{equation}
u=q^{2n}s^{\pm 1},\quad n\in\mathbb{Z},
\end{equation}
which all are simple provided $s\not\in \pm q^{\Z}$. %The second order zeros are the subject of separate investigation.
The numerator of (\ref{final-ansatz}) vanishes for any $u$ satisfying the equation
\begin{equation}\label{cancellation}
G_q(u,\energy )=-\xi\overline{G_q(u,\energy )}.
\end{equation}
On the other hand, the identity $W(q^{2n}s^{\pm 1})=0$ is equivalent to
%\begin{equation}
%q^2\frac{\bchi_q(q^{2n}s,\energy )}{\bchi_q(q^{-2n}s^{-1},\energy )}\;=\;\frac{\bchi_q(q^{2n-2}s,\energy )}{\bchi_q(q^{2-2n}s^{-1},\energy )}\;,
%\end{equation}
%what is equivalent to
\begin{equation}\label{GG}
G_q(q^{2n}s,\energy )=G_q(s,\energy )
\end{equation}
which means that the infinite number of cancellation conditions arising from (\ref{cancellation}) at $u=q^{2n} s^{\pm1}$
%\begin{equation}
%G_q(q^{2n}s^{\pm1},\energy ) =-\xi \overline{G(q^{2n} s^{\pm1},\energy )},\quad \forall n\in\mathbb{Z},
%\end{equation}
reduce to the  single equation~\eqref{spectral}.
%\begin{equation}
%G_q(s,\energy )\;=\;-\xi\overline{G_q(s,\energy )}.
%\end{equation}
\end{proof}
\noindent\textbf{Remark.} The spectral equation~\eqref{spectral} together with the  (anti)symmetry property~\eqref{antisym}
%\begin{equation}
%G_q(u^{-1},\energy )\;=\;G_q(u,\energy )^{-1}
%\end{equation}
implies  that
\begin{equation}
\xi^2 = 1
\end{equation}
regardless of our previous parity and reality considerations.

\section{Numerical study for $\hbar=2\pi\ii$}\label{sec5}

\subsection{The structure of $\energy =\energy (\sigma)$}

Let us remind the definition of the Wronskian (\ref{wron}):
\begin{equation}\label{Wron2}
W(u,\energy ):=\bchi\left(\frac{u}{q^2},\energy \right)\check\bchi(u,\energy )-\check\bchi\left(\frac{u}{q^2},\energy \right)\bchi(u,\energy )=\varrho(\energy )\theta_1(su)\theta_1(s^{-1}u).
\end{equation}
Let $\sigma$ be fixed ($s=\EXP^{2\pi\bb\sigma}$). Consider the following equation relating $\sigma$ and  $\energy $:
\begin{equation}\label{Eeq}
W(\EXP^{2\pi\bb\sigma},\energy )=0.
\end{equation}
We have used the Newton method to solve (\ref{Eeq}) with respect to $\energy $. This method uses an initial point as input, and different initial points provide different solutions:
\begin{equation}
\energy \in\{\energy _1,\energy _2,\energy _3,\dots\}
\end{equation}
Varying $\sigma$, one obtains
\begin{equation}
\energy (\sigma)\in\{\energy _1(\sigma),\energy _2(\sigma),\energy _3(\sigma),\dots\},
\end{equation}
where $\energy _k(\sigma)$ stands for $k^{th}$ sheet of the Riemann surface of the multivalued function $\energy (\sigma)$.

If $\energy $ and $\sigma$ are related by (\ref{Eeq}), then one has
\begin{equation}
W(\EXP^{2\pi\bb(\pm \sigma + \ii \bb n + \ii \bb^{-1}m)},\energy )=0,\quad \forall n,m\in\mathbb{Z}.
\end{equation}
We deduce that
\begin{equation}
\energy (\sigma)=\energy (\pm\sigma + \ii \bb n +\ii\bb^{-1} m).
\end{equation}
In general, this relation involves exchanges of sheets of the Riemann surface. However, more careful numerical studies show that 
\begin{itemize}
\item quantisation corresponds to only real $\sigma$;
\item all branching points of $\energy =\energy (\sigma)$ lie outside the real axis.
\end{itemize}
Numerical studies show that indeed 
\begin{equation}
\energy _k(\sigma)=\energy _k(-\sigma)=\energy _k(2\sin\theta-\sigma),\quad \sigma\in\mathbb{R},
\end{equation}
for the same $k^{th}$ sheet of the Riemann surface. In what follows, the image of the function $\energy _k(\sigma)$ with $\sigma\in\mathbb{R}$ will be called \emph{orbit}. The principal domain for $\energy _k(\sigma)$ is then
\begin{equation}
\sigma\in\mathbb{R},\qquad 0\leq\sigma\leq\sin\theta.
\end{equation}

In the numerical studies we use $\bb=\EXP^{\frac{\ii\pi}{4}}$ which corresponds to
\begin{equation}
\theta=\frac{\pi}{4},\quad \sin\theta=\cos\theta=\frac{1}{\sqrt{2}},
\quad
q=\EXP^{-\pi}.
\end{equation}

Figs~\ref{fig1} and  \ref{fig2} show the orbits of $\energy _1(\sigma)$ and  $\energy _2(\sigma)$ respectively in the complex plane. 
Two  branching points for the branch cuts between $\energy _1$ and $\energy _2$ are located approximately at
\begin{equation}
\sigma \approx \frac{1}{\sqrt{2}} \pm 0.007 \bb.
\end{equation}

The right hand side of Fig.~\ref{fig2} shows the gap between $\energy _1$ and $\energy _2$. Fig.~\ref{fig3} shows the orbit $\energy _3(\sigma)$. The right hand side of Fig.~3 shows all three orbits in the log scale.

\begin{figure}[ht]
\begin{center}
\includegraphics*[scale=0.4]{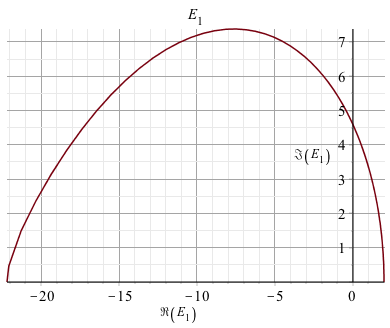}
\caption{The orbit of $\energy =\energy _1(\sigma)$. Turning points (endpoints on the plot) are $\energy _1(0)=1.9962511523$ and $\energy _1(\sin\theta)=-22.1838257068$.}
\label{fig1}
\end{center}
\begin{center}
\includegraphics*[scale=0.4]{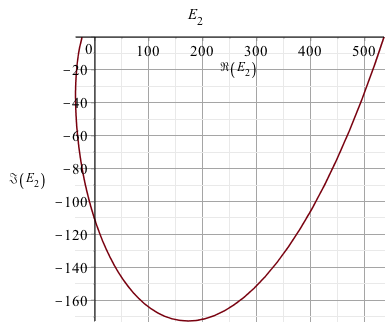}\hskip 1cm
\includegraphics*[scale=0.4]{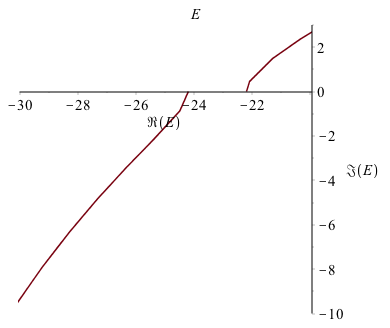}
\caption{Left: the orbit of $\energy =\energy _2(\sigma)$. Turning points are $\energy _2(0)=535.49351947$ and $\energy _2(\sin\theta)=-24.183825694$. Right: the orbits of $\energy _1(\sigma)$ and $\energy _2(\sigma)$ in the vicinity of $\sigma=1/\sqrt{2}$. Recall that $\energy _1(1/\sqrt{2})=-22.1838257068$ vs. $\energy _2(1/\sqrt{2})=-24.183825694$.}
\label{fig2}
\end{center}
\begin{center}
\includegraphics*[scale=0.4]{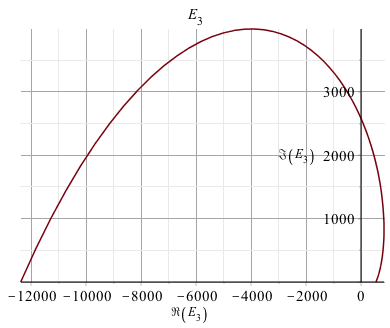}\hskip 1cm
\includegraphics*[scale=0.4]{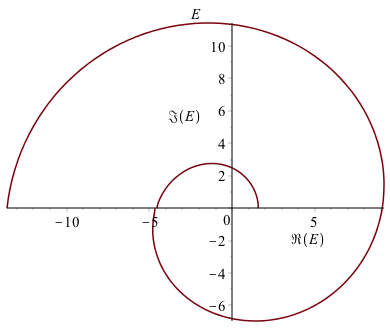}
\caption{Left: the orbit of $\energy =\energy _3(\sigma)$. Turning points are $\energy _3(0)=535.49726832$ and $\energy _3(\sin\theta)=-12391.6479693$. Right: Collection of $\energy _1(\sigma)$, $\energy _2(\sigma)$ and $\energy _3(\sigma)$ in the logarithmic scale. Low resolution does not allow one to see the gaps between the orbits. Approximation of this spiral trajectory is $\energy (\sigma)\simeq \EXP^{2\pi\bb\sigma}$, $\sigma\geq 0$.}
\label{fig3}
\end{center}
\end{figure}
It is worth discussing the endpoints $\sigma=0$ and $\sigma=\sin\theta$ in general. The point $\sigma=0$ gives (see (\ref{Wron2}))
\begin{equation}
W(1,\energy )=\bchi(1,\energy )\left( \bchi(q^{-2},\energy )-q^2\bchi(q^2,\energy )\right)=0,
\end{equation}
and the case
\begin{equation}
\bchi(1,\energy )=0
\end{equation}
gives $\energy _k(0)$ for even sheets\footnote{It suffices to substitute $\ds \energy =q^{-2k}\sum_{n=0}^\infty c_nq^{2n}$ into $\bchi(1)$ as a series with respect to $q^2$ and equate the coefficients of each power $q^{2n}$ to zero.},
\begin{equation}
\begin{array}{l}
\ds \energy _2(0)=\frac{1}{q^2} \left( 1+q^4-q^6-q^8-q^{10}-2q^{12}-q^{14}+q^{18}+6q^{20}+11q^{22}+\dots\right),\\
[4mm]
\ds \energy _4(0)=\frac{1}{q^4} \left( 1+2q^8+q^{14}-q^{18}-3q^{20}-8q^{22}-13q^{24}+\dots\right),\\
[4mm]
\ds \energy _6(0)=\frac{1}{q^6} \left( 1+2q^{12}+q^{22}+q^{24}+q^{26}+2q^{32}+2q^{34}+\dots\right),
\end{array}
\end{equation}
and so on. Another branch 
\begin{equation}
\bchi(q^{-2},\energy )-q^2\bchi(q^2,\energy )=0
\end{equation}
gives $\energy _k(0)$ for the odd sheets:
\begin{equation}
 \energy _1(0)
 =2-2q^2-4q^4-2q^6+14q^8+50q^{10}+40q^{12}-268q^{14}-1136q^{16}+\dots,
\end{equation}
\begin{equation}
 \energy _3(0)
 =\frac{1}{q^2} (1+3q^4+3q^6+q^8-15q^{10}-52q^{12}-43q^{14}+264q^{16}+1127q^{18}+\dots),
\end{equation}
and so on.

The point $\sigma=\sin\theta$ gives 
\begin{equation}
W(-q^{-1},\energy )\sim\left(\bchi(-q^{-1},\energy )-q\bchi(-q,\energy )\right)\left(\bchi(-q^{-1},\energy )+q\bchi(-q,\energy )\right)
\end{equation}
so that  
\begin{equation}\label{minus}
\bchi(-q^{-1},\energy )-q\bchi(-q,\energy )=0
\end{equation}
gives the values for $\energy _k(\sin\theta)$ on the odd sheets, for instance
\begin{equation}
\energy _1(\sin\theta)=-q^{-1}\left( 1-q+q^2-q^4-q^6+q^7-2q^8+2q^9-2q^{10}+5q^{11}-4q^{12}+\dots\right).
\end{equation}
Equation 
\begin{equation}\label{plus}
\bchi(-q^{-1},\energy )+q\bchi(-q,\energy )=0
\end{equation}
gives the values for $\energy _k(\sin\theta)$ on the even sheets, for instance
\begin{equation}
\energy _2(\sin\theta)=-q^{-1}\left( 1+q+q^2-q^4-q^6-q^7-2q^8-2q^9-2q^{10}-5q^{11}-4q^{12}-\dots\right).
\end{equation}

\subsection{The spectrum}
We turn now to the spectrum of the Hamiltonians (\ref{Hams}).

The quantisation conditions (\ref{spectral}) for real $\sigma$ are
\begin{equation}
\textrm{Im}(G_q(s,\energy ))=0,\qquad \xi=-1,\quad \textrm{odd state,}
\end{equation}
and
\begin{equation}
\textrm{Re}(G_q(s,\energy ))=0,\qquad \xi=1,\quad \textrm{even state.}
\end{equation}

One can see that formally there are two points $\sigma=0$ and $\sigma=\sin\theta$ on each sheet of the Riemann surface of $\energy(\sigma)$ corresponding to these quantisation conditions. These are the exceptional cases as  denominator $W(u)$ has the second order zeros while in the cases $\sigma=0$ for odd sheets, and $\sigma=\sin\theta$ for all sheets, the numerator of (\ref{final-ansatz}) has the first order zeros. Therefore, all these cases must be eliminated from the spectrum. On the other hand,  if $\sigma=0$ on even sheets, corresponding to $\bchi(1)=0$, the numerator of (\ref{final-ansatz}) has second order zeros at some second order zeros of the denominator, but not at all of them and therefore the wave function (\ref{final-ansatz}) could not be a holomorphic function either. Theorem~\ref{spectral-theorem} concerns only the elimination of simple poles, and further analysis is needed to completely understand the situation with double poles.

\subsubsection{Sheet 1.}
Real and imaginary parts of $G_q(s,\energy _1(\sigma))$ (the first sheet) are plotted in Fig.~\ref{fig4}.

\begin{figure}[ht]
\begin{center}
\includegraphics*[scale=0.4]{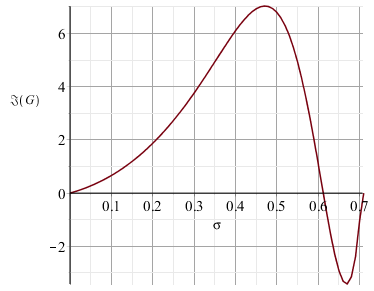}\hskip 1cm
\includegraphics*[scale=0.4]{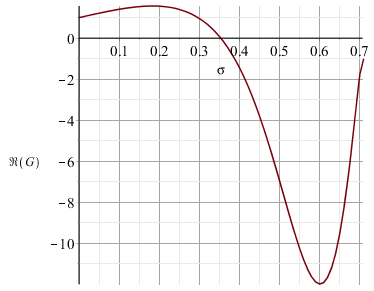}
\caption{Left: $\Im(G_q(s,\energy _1))$ as function of $\sigma$, $0\leq\sigma\leq\sin\theta$. Right: the same for $\Re(G_q(s,\energy _1))$.}
\label{fig4}
\end{center}
\end{figure} 
An even state from the first sheet is given by
\begin{equation}
\sigma = \frac{1}{2}\sin\theta,\quad  \energy _1(\sigma) = 4.594 358 809 836 918 94 \ii.
\end{equation}
This is the ground state.

An odd state from the first sheet is given by
\begin{equation}
\sigma=0.612 117 371 646 167 267 5 , \quad\energy _1(\sigma) = -13.878 304 778 036 690 6 +6.161 296 243 244 348 685\ii.
\end{equation}
The cases corresponding to $\sigma=0$ and $\sigma=\sin\theta$ are discarded.

\subsubsection{Sheet 2.} Imaginary and real parts of the function $\ds \frac{G_q(s,\energy _2(\sigma))}{|G_q(s,\energy _2(\sigma))|}$ (on the second sheet) are plotted in Fig.~\ref{fig5}.

\begin{figure}[ht]
\begin{center}
\includegraphics*[scale=0.4]{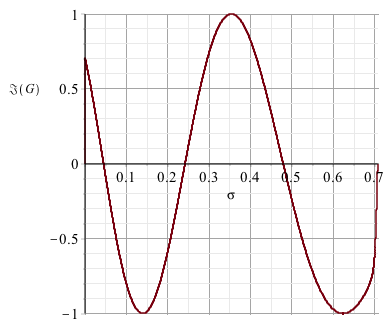}\hskip 1cm
\includegraphics*[scale=0.4]{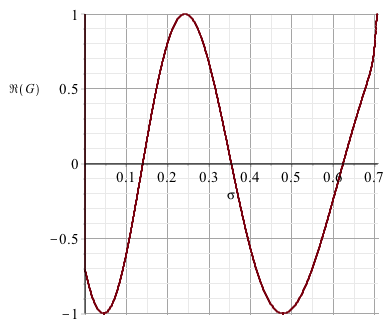}
\caption{Left: $\ds\frac{\Im(G_q(s,\energy _2))}{|G_q(s,\energy _2)|}$ as function of $\sigma$, $0\leq\sigma\leq \sin\theta$. Right: the same for $\ds\frac{\Re(G_q(s,\energy ))}{|G_q(s,\energy _2)|}$. Behaviour near $\sigma=0$ is misleading since one has an ambiguity $0/0$ there.}
\label{fig5}
\end{center}
\end{figure} 

\begin{table}[h]
\[
 \begin{array}{|c|c|}
 \hline
  \sigma& \energy_2(\sigma)\\
  \hline
 0.044 907 405 413 666 898 6& 429. 937 612 699 070 933 - 86.935 286 983 923 622 8\ii\\
 0.241 612 973 133 940 861&  87.333 249 871 603 300 85 - 160.859 744 733 070 428 \ii\\
 0.478 766 031 821 187 121& -33.715 476 768 740 864 9 - 54.171 056 749 691 862 2\ii\\
 \hline
  \end{array}
\]
\caption{The odd states from the second sheet.}
\label{tb1}
\end{table}
The odd and even states from the second sheet are given in Tables~\ref{tb1} and \ref{tb2} respectively.
\begin{table}[h]
\[
 \begin{array}{|c|c|}
 \hline
  \sigma& \energy_2(\sigma)\\
  \hline
0.139 460 116 715 428 804& 234.614 101 715 470 239  - 167.345 168 305 129 794 \ii\\
 \frac{1}{2}\sin\theta&   -111.300 184 113 096 796 \ii\\
  0.623 413 635 048 467 267&   -31.325 044 899 672 604 73 - 12.153 338 942 267 673 58\ii\\
 \hline
  \end{array}
\]
\caption{The even states from the second sheet.}
\label{tb2}
\end{table}
The endpoint on the second orbit is given by
\begin{equation}
\sigma = 0,\quad \energy _2(0)=535. 493 519 473 629 469,
\end{equation}
corresponding to $\bchi(1,\energy_2)=0$, but this value does not correspond to an eigenvalue since it is not possible  to cancel all second order poles in the wave function.

\subsection{Table of the spectrum for $\bb=\EXP^{\ii\pi/4}$}
We summarize the spectrum from the first and second sheets in Tables~\ref{tb3} and  \ref{tb4} respectively.
In our numerics we have used the precision $10^{-40}$, therefore all our digits for the energies are precise.
\begin{table}[h]
\[
 \begin{array}{|c|c|c|}
 \hline
  \sigma& \text{parity}& \energy_1(\sigma)\\
  \hline
0.35& \text{even}&4.594 358 809 836 918 94 \ii\\
0.61&\text{odd}&    -13.878 304 778 036 690 6 +6.161 296 243 244 348 685\ii\\
 \hline
  \end{array}
\]
\caption{The states from the first sheet, $k=1$. }
\label{tb3}
\end{table}
\begin{table}[h]
\[
 \begin{array}{|c|c|c|}
 \hline
  \sigma& \text{parity}& \energy_2(\sigma)\\
  \hline
0.62& \text{even}& -31.325 044 899 672 604 73 - 12.153 338 942 267 673 58\ii\\
0.48&\text{odd}&  -33.715 476 768 740 864 9 - 54.171 056 749 691 862 2\ii\\
0.35&\text{even}& -111.300 184 113 096 796 \ii\\
0.24&\text{odd}& 87.333 249 871 603 300 85 - 160.859 744 733 070 428 \ii\\
0.14&\text{even}&234.614 101 715 470 239  - 167.345 168 305 129 794 \ii\\
0.045&\text{odd}& 429. 937 612 699 070 933 - 86.935 286 983 923 622 8\ii\\
%0& \text{even and odd}& 535. 493 519 473 629 469\\
 \hline
  \end{array}
\]
\caption{The states from the second sheet, $k=2$. }
\label{tb4}
\end{table}

\section{The self-conjugate case $\hbar=2\pi$}\label{sec6}
Despite the fact that the special case $\bb=1$ corresponds to the limiting value of the strong coupling regime $\theta\to0$, we do not know how to approach that limit on the level of eigenvalues and eigenfunctions. For this reason, we treat that case separately, independently of the strong coupling regime. This case is very special as the operators $\uop$ and $\vop$ are positive self-adjoint and formally commuting, while both functional equations~(\ref{DEs}), \eqref{DEsc} reduce to one and the same but still non-trivial functional difference equation
\begin{equation}\label{psi-b=1}
\psi(x-\ii)+\psi(x+\ii)=(\energy -2\cosh(2\pi x))\psi(x),
\end{equation}
which, under the substitution
\begin{equation}
\varphi(x):=\psi(\ii x),
\end{equation}
becomes an equation of the form
\begin{equation}
\label{phi-b=1}
\varphi(x-1)+\varphi(x+1)+(2\cos(2\pi x)-\energy )\varphi(x)=0.
\end{equation}
Despite being a special case of Harper's functional equation~\cite{0370-1298-68-10-304} on the real line, this equation should nonetheless be considered on the entire complex plane as its quantum mechanical content is extracted from the restriction of its solutions to the imaginary axis. 
\subsection{The spectral curve}
Due to periodicity of the potential, for any solution of \eqref{phi-b=1}, its Weil--Gel'fand--Zak transform \cite{MR0005741,MR0039154,Zak1967}
\begin{equation}\label{wgzt}
\tilde \varphi(x,y)=\sum_{m\in\Z}\varphi(x+m)\EXP^{-2\pi\ii ym}
\end{equation}
 gives another solution which is quasi-periodic in $x$ and periodic in $y$:
 \begin{equation}\label{q-per}
\tilde \varphi(x+1,y)=\EXP^{2\pi\ii y}\tilde \varphi(x,y),\quad \tilde \varphi(x,y+1)=\tilde \varphi(x,y).
\end{equation}
That means that the functional difference equation~\eqref{phi-b=1} is transformed into an algebraic functional equation
\begin{equation}\label{alg-f-e}
\left(\cos(2\pi x)+\cos(2\pi y)-\frac{\energy }2\right)\tilde \varphi(x,y)=0,
\end{equation}
which means that the support of $\tilde \varphi(x,y)$ is localised to the \emph{spectral curve}
\begin{equation}
\Gamma_\energy :=\left\{(x,y)\in\C^2\left \vert\  \cos(2\pi x)+\cos(2\pi y)=\frac{\energy }2\right.\right\}
\end{equation}
which covers the finite (affine) part of the elliptic curve
\begin{equation}
\digamma_\energy :=\left\{(u,v)\in\C^2\left \vert\  u+u^{-1}+v+v^{-1}=\energy \right.\right\}
\end{equation}
through the map
\begin{equation}
\operatorname{q}\colon \Gamma_\energy \to \digamma_\energy ,\quad (x,y)\mapsto\left(\EXP^{2\pi\ii x},\EXP^{2\pi\ii y}\right).
\end{equation}
Projection to the first component
\begin{equation}
x\colon \Gamma_\energy \to \C
\end{equation}
is a covering map branched at the points satisfying the condition $\sin(2\pi y)=0$ which corresponds to a discrete infinite set of values of $x$:
\begin{equation}\label{brpts}
(-\ii\alpha+\Z)\cup(\ii\alpha+\Z)\cup(-\ii\beta+\Z)\cup(\ii\beta+\Z)
\end{equation}
where parameters $\alpha=\alpha(\energy )$ and $\beta=\beta(\energy )$ are determined through the equations
\begin{equation}\label{aleq}
\energy =2+2 \cosh(2\pi\alpha)=4\cosh(\pi\alpha)^2,
\end{equation}
\begin{equation}\label{beeq}
\cosh(2\pi\beta)=\cosh(2\pi\alpha)+2\quad\Leftrightarrow\quad \sinh(\pi\beta)=\pm\cosh(\pi\alpha),
\end{equation}
%and the map $D_\energy \colon \C\to\C$, which we will call \emph{discriminant}, is defined by the formula
%\begin{multline}
%D_\energy (x):= \left(\frac{\energy }2-\cos(2\pi x)\right)^2-1\\
%= 4 \sin(\pi(x+\ii\alpha))  \sin(\pi(x-\ii\alpha)) \sin(\pi(x+\ii\beta)) \sin(\pi(x-\ii\beta)).
%\end{multline}
%It is easily verified that
%\begin{equation}
%x^*(D_\energy )=-\sin(2\pi y)^2
%\end{equation}
Since the spectrum of our operator is contained in $\R_{>4}$, we assume that $\energy \in\R_{>4}$. In that case, we can choose $\alpha$ and $\beta$ in a canonical way as real positive solutions of equations~\eqref{aleq} and \eqref{beeq}.

The automorphism group of the spectral curve $\Gamma_\energy $ contains a subgroup generated by  three elements
\begin{equation}\label{auts}
\rho,\sigma,\tau \colon \Gamma_\energy \to \Gamma_\energy ,\quad \rho(x,y)=(-x,y),\quad \sigma(x,y)=(y,x),\quad \tau(x,y)=(x+1,y).
\end{equation}
The two involutions $\rho$ and $\sigma$ cover involutions $\dot\rho$ and $\dot\sigma$ of $\digamma_\energy$, while $\tau$ acts fiber-wise, i.e.
\begin{equation}
\operatorname{q}\circ\rho=\dot\rho\circ\operatorname{q},\quad \operatorname{q}\circ\sigma=\dot\sigma\circ\operatorname{q},\quad \operatorname{q}\circ\tau=\operatorname{q}.
\end{equation}
\subsection{The deformed symplectic potential}
For the construction of a Bloch--Jost function to be defined below, we will use the restrictions to $\Gamma_\energy $ of a one parameter family of \emph{deformed (holomorphic) symplectic potentials} on $\C^2$
\begin{equation}
\theta_\lambda:=\left(x+\frac{\lambda}{\sin(2\pi x)}\right)\operatorname{d}\! y,\quad \lambda\in\C,
\end{equation}
where the non-deformed part
\begin{equation}
\theta_0=x\operatorname{d}\! y
\end{equation}
is the restriction of the standard holomorphic symplectic potential in $\C^2$, while the deforming part is given by the pull-back of a holomorphic 1-form on  $\digamma_\energy$
\begin{equation}
\frac{\operatorname{d}\! y}{\sin(2\pi x)}=-\frac{\operatorname{d}\! x}{\sin(2\pi y)}=\operatorname{q}^*(\omega),\quad
\omega=\frac{\operatorname{d}\! v}{\pi v(u-u^{-1})}=
-\frac{\operatorname{d}\! u}{\pi u(v-v^{-1})},
\end{equation}
 which is unique up to multiplication by a non-zero complex constant.   With respect to the automorphisms~\eqref{auts}, the form $\theta_\lambda$ has the following properties:
 \begin{equation}\label{tltr}
\rho^*(\theta_\lambda)=-\theta_\lambda,\quad \sigma^*(\theta_\lambda)=\operatorname{d}(xy)-\theta_\lambda,\quad
\tau^*(\theta_\lambda)=\operatorname{d}\! y+\theta_\lambda.
\end{equation}

\subsection{The Bloch--Jost function}
\begin{definition}
%Let $\operatorname{p}\colon\tilde\Gamma_\energy\to \Gamma_\energy$ be the universal covering.
 A holomorphic function $f\colon \Gamma_\energy \to \C$ is called \emph{Bloch--Jost function} if it satisfies the functional relations
 \begin{equation}
f(x+1,y)\EXP^{-2\pi\ii y}=f(x,y+1)=f(x,y),\quad \forall (x,y)\in \Gamma_\energy .
\end{equation}
\end{definition}

For a fixed base point $P_0\in \Gamma_\epsilon$, let $\operatorname{p}\colon\tilde\Gamma_\energy\to \Gamma_\energy$ be the universal covering given by the homotopy classes of all paths starting from $P_0$. We define
\begin{equation}
g_{\lambda,\energy}\colon \tilde\Gamma_\energy\to \C_{\ne0},\quad \gamma\mapsto \EXP^{2\pi\ii\int_\gamma\theta_\lambda}
\end{equation}

\begin{theorem}\label{bl-jothm}
 Let $(\lambda,\energy )\in \C\times \R_{>4}$ be such that
\begin{equation}\label{speceqns}
g_{\lambda,\energy}(\gamma)=1,
\end{equation}
for any $\gamma$ for which  $x\circ \gamma$  and $\operatorname{q}\circ\gamma$ are closed paths. Then, there exists a Bloch--Jost function $f$ such that $g_{\lambda,\energy}=\operatorname{p}^*(f)$.
\end{theorem}
\begin{proof}
The fact that $g_{\lambda,\energy}=\operatorname{p}^*(f)$ for some $f\colon \Gamma_\energy \to \C$ is immediate. It suffices to remark that any closed path on $\Gamma_\energy $ projects down to closed paths both in $\C$ and $\digamma_\energy $. So, we need to verify only the quasi-periodicity properties of $f$.
 
 Let  paths $\mu, \nu\colon [0,1]\to \Gamma_\energy $ be such that
\begin{equation}
\mu(0)=P_0,\quad \mu(1)=\nu(0)=(x,y),\quad \nu(1)=(x,y+1).
\end{equation}
Then, we have
\begin{equation}
f(x,y+1)=\EXP^{2\pi\ii \int_{\mu\cdot\nu}\theta_\lambda}=\EXP^{2\pi\ii \int_{\mu}\theta_\lambda}\EXP^{2\pi\ii \int_{\nu}\theta_\lambda}=f(x,y),
\end{equation}
since both $x\circ \nu$  and $\operatorname{q}\circ\nu$ are closed.

Redefine now $\mu, \nu\colon [0,1]\to \Gamma_\energy $ so that
\begin{equation}
\mu(0)=P_0,\quad \mu(1)=\nu(0)=(x,y),\quad \nu(1)=(x+1,y).
\end{equation}
Then, we have
\begin{equation}
f(x+1,y)=\EXP^{2\pi\ii \int_{\mu\cdot\nu}\theta_\lambda}=\EXP^{2\pi\ii \int_{\mu}\theta_\lambda}\EXP^{2\pi\ii \int_{\nu}\theta_\lambda}=f(x,y)\EXP^{2\pi\ii \int_{\nu}\theta_\lambda},
\end{equation}
where, this time, the path $x\circ \nu$ is not closed. Using the secon equality in \eqref{tltr},
we have
\begin{equation}
\EXP^{2\pi\ii\int_{\nu}\theta_\lambda}=\EXP^{2\pi\ii\int_{\sigma\circ\nu}\sigma^*\theta_\lambda}
=\EXP^{2\pi\ii\int_{\sigma\circ\nu}\operatorname{d}(xy)}\EXP^{-2\pi\ii\int_{\sigma\circ\nu}\theta_\lambda}
=\EXP^{2\pi\ii y}
\end{equation}
where we have used the fact that the paths $x\circ\sigma\circ \nu$  and $\operatorname{q}\circ\sigma\circ\nu$ are closed.
\end{proof}
\subsection{The spectral equations}
We call equations~\eqref{speceqns} \emph{spectral equations}. Our goal here is to reduce them to a minimal set so that any other equation is a consequence of the equations from that minimal set. By using two functions 
\begin{equation}
r\colon \R\to [\alpha,\beta],\quad \cosh(2\pi r(t))=1-\cos(\pi t)+\cosh(2\pi\alpha),
\end{equation}
and 
\begin{equation}
s\colon \R\to [0,\alpha],\quad \sinh(\pi s(t))=\sinh(\pi\alpha)\sin(\pi t),
\end{equation}
 we define  four paths $\xi,\check\xi,\zeta,\hat\zeta\colon [0,1]\to \Gamma_\energy$ by
\begin{equation}\label{xi}
\xi(t)=(\ii s(t+1/2),\ii s(t)),\quad \check\xi(t)=(-\ii s(t+1/2),-\ii s(t)).
\end{equation}
and
\begin{equation}\label{zeta}
\zeta(t)=(\ii r(t),t/2),\quad \hat\zeta(t)=(\ii r(1-t),(t+1)/2).
\end{equation}
We observe that $(\xi,\check\xi)$ and $(\zeta,\hat\zeta)$ are composable pairs, $\xi\cdot\check\xi$ being closed, while $\zeta\cdot \hat\zeta$ projects down to closed paths both in $\C$ and $\digamma_\energy$. Moreover, their images generate the fundamental group of $\digamma_\energy$, which means that the equations
%\begin{equation}
%\xi^*\theta_\lambda
%=\frac\ii2\left(r(t)- \frac{\lambda}{\sinh(2\pi r(t))}\right)\operatorname{d}\! t,
%\end{equation}
%The  equation
\begin{equation}\label{speceqs1}
\EXP^{2\pi\ii \int_{\xi\cdot\check\xi}\theta_\lambda}=\EXP^{2\pi\ii \int_{\zeta\cdot\hat\zeta}\theta_\lambda}=1
\end{equation}
constitute a complete set of spectral equations.
We calculate
\begin{equation}\label{int1=0}
\int_{\xi\cdot\check\xi}\theta_\lambda=2\int_{\xi}\theta_\lambda=4\int_{0}^{1/2}\left( \frac{\lambda}{\sinh(2\pi s(t+1/2))}-s(t+1/2)\right)s'(t)\operatorname{d}\!t=:A\lambda-\tilde A
\end{equation}
and 
\begin{equation}\label{int1=1}
\int_{\zeta\cdot\hat\zeta}\theta_\lambda=2\int_{\zeta}\theta_\lambda=\ii\int_{0}^{1}\left(r(t)- \frac{\lambda}{\sinh(2\pi r(t))}\right)\operatorname{d}\!t=:\ii(\tilde B-\lambda B)
\end{equation}
where $A, \tilde A, B,\tilde B$ are positive functions of $\energy$ satisfying the inequality
\begin{equation}\label{abba}
A\tilde B-B\tilde A>0
\end{equation}
which is seen if we write out all integrals explicitly:
\begin{equation}
 A\tilde B-B\tilde A=4\int_{[0,1/2]\times[0,1]}\operatorname{d}(t,u)\frac{s'(t)\left(r(u)\sinh(2\pi r(u))- s(t+1/2)\sinh(2\pi s(t+1/2))\right)}{\sinh(2\pi s(t+1/2))\sinh(2\pi r(u))}
\end{equation}
and take into account the inequalities
\begin{equation}
0\le s(t)\le\alpha\le r(u)\le\beta,\quad \forall t,u\in\R.
\end{equation}
Now, the spectral equations~\eqref{speceqs1} imply that 
\begin{equation}\label{qc1}
A\lambda-\tilde A=:n+1\in\Z, \quad \ii(\tilde B-\lambda B)\in\Z,
\end{equation}
so that the second condition fixes a unique choice for $\lambda$,
\begin{equation}\label{lamb}
\lambda=\tilde B/B,
\end{equation}
while the first condition becomes a quantisation condition for the energy:
\begin{equation}
A\tilde B-B\tilde A =(n+1)B,\quad n\in\Z_{\ge0},
\end{equation}
where we have taken into account the inequality~\eqref{abba}. 

Numerical calculation at $n=0$ for the ground state  energy $\energy _0$ gives the result
\begin{equation}
\log(\energy  _0)=2.88181542992629678247713987172363292221616219\dots
\end{equation}
which is in agreement with formula~(5.39) of \cite{MR3556519}.
\subsection{The eigenfunctions}
\begin{theorem}
 Let $(\lambda,\energy )\in \C\times \R_{>4}$ satisfy the assumptions of Theorem~\ref{bl-jothm}, and choose the base point  $P_0=(\ii\alpha,0)\in\Gamma_\energy $. Then, there  exists a holomorphic function $\varphi \colon\C\to\C$  such that
\begin{equation}\label{soln}
\sin(2\pi y)\varphi(x)=\sin\left(2\pi\int_{\gamma}\theta_\lambda\right),\quad \forall (x,y)\in \Gamma_\energy ,
\end{equation}
for any path $\gamma \colon [0,1]\to \Gamma_\energy $ such that $\gamma(0)=P_0$ and $\gamma(1)=(x,y)$. In particular, $\varphi(x)$ is a solution of the functional difference equation~\eqref{phi-b=1}.
\end{theorem}
\begin{proof}
 The right hand side of \eqref{soln} can be rewitten in the form
 \begin{equation}
2\ii \sin\left(2\pi\int_{\gamma}\theta_\lambda\right)=f(x,y)-f(x,-y),
\end{equation}
where $f(x,y)$ is the Bloch--Jost function defined by
\begin{equation}
f(x,y):=\EXP^{2\pi\ii\int_\gamma\theta_\lambda}
\end{equation}
and the choice of the base point is such that
\begin{equation}
f(x,y)f(x,-y)=1.
\end{equation}
Thus, the function
\begin{equation}
h(x,y):=\frac{f(x,y)-f(x,-y)}{\sin(2\pi y)}
\end{equation}
is explicitly periodic and symmetric in $y$:
\begin{equation}
h(x,y+1)=h(x,y)=h(x,-y)
\end{equation}
On the other hand, if $(a,b)\in\Gamma_\energy $ then
\begin{equation}
x^{-1}(a)=\{a\}\times (\{b+\Z\}\cup\{-b+\Z\})
\end{equation}
so that
\begin{equation}
h(a,b')=h(a,b),\quad \forall (a,b')\in x^{-1}(a),
\end{equation}
i.e. there exists a function $\varphi\colon\C\to\C$ such that $h=x^*(\varphi)$. It remains to see that $\varphi(x)$ is holomorphic at any of the ramification points \eqref{brpts} corresponding to $\sin(2\pi y)=0$, i.e. either $y=0$ or $y=1/2$. The case $y=0$ is explicitly regular. To see the regularity at $y=1/2$, it suffices to see that 
\begin{equation}\label{cancellation-cond}
f(\pm\ii\beta,1/2)^2=1.
\end{equation}
In the case of plus sign, from the definition~\eqref{zeta}, it follows that
\begin{equation}\label{cancellation-cond+}
f(\ii\beta,1/2)=\EXP^{2\pi\ii\int_\zeta\theta_\lambda}=1
\end{equation}
where the last equality is satisfied due to \eqref{int1=0} and \eqref{lamb}.

In the case of minus sign in \eqref{cancellation-cond}, we verify it by using \eqref{cancellation-cond+} and the symmetry property of the Bloch--Jost function under the negation of the first component: %the definition~\eqref{xi} we have 
\begin{equation}
f(-x,y)=f(-\ii\alpha,0)/f(x,y)=\EXP^{2\pi\ii\int_\xi\theta_\lambda}/f(x,y)=(-1)^{n+1}/f(x,y)
\end{equation}
where we have used definition~\eqref{xi} and the equation in the first part of \eqref{qc1}.
\end{proof}
Thus, we have  solved the spectral problem~\eqref{psi-b=1}, where the spectrum is determined by equations~\eqref{qc1} and the corresponding eigenfunctions are given by the formula $\psi(x)=\varphi(\ii x)$ for all $x\in \R$, where $\varphi(x)$ is defined in \eqref{soln}. These eigenfunctions coincide with the ones constructed in \cite{MarinoZakany2016}.
%\bibliography{/Users/rinatkashaev/Dropbox/LatexStaff/biblio}{}

\begin{thebibliography}{10}

\bibitem{MR2191887}
Mina Aganagic, Robbert Dijkgraaf, Albrecht Klemm, Marcos Mari\~no, and Cumrun
  Vafa.
\newblock Topological strings and integrable hierarchies.
\newblock {\em Comm. Math. Phys.}, 261(2):451--516, 2006.

\bibitem{MR690578}
Rodney~J. Baxter.
\newblock {\em Exactly solved models in statistical mechanics}.
\newblock Academic Press, Inc. [Harcourt Brace Jovanovich, Publishers], London,
  1982.

\bibitem{MR1832065}
Vladimir~V. Bazhanov, Sergei~L. Lukyanov, and Alexander~B. Zamolodchikov.
\newblock Spectral determinants for {S}chr\"odinger equation and {${\bf
  Q}$}-operators of conformal field theory.
\newblock In {\em Proceedings of the {B}axter {R}evolution in {M}athematical
  {P}hysics ({C}anberra, 2000)}, volume 102, pages 567--576, 2001.

\bibitem{MR1733841}
Patrick Dorey and Roberto Tateo.
\newblock Anharmonic oscillators, the thermodynamic {B}ethe ansatz and
  nonlinear integral equations.
\newblock {\em J. Phys. A}, 32(38):L419--L425, 1999.

\bibitem{MR1345554}
L.~D. Faddeev.
\newblock Discrete {H}eisenberg--{W}eyl group and modular group.
\newblock {\em Lett. Math. Phys.}, 34(3):249--254, 1995.

\bibitem{MR0039154}
I.~M. Gel{\cprime}fand.
\newblock Expansion in characteristic functions of an equation with periodic
  coefficients.
\newblock {\em Doklady Akad. Nauk SSSR (N.S.)}, 73:1117--1120, 1950.

\bibitem{MR3556519}
Alba Grassi, Yasuyuki Hatsuda, and Marcos Mari\~no.
\newblock Topological strings from quantum mechanics.
\newblock {\em Ann. Henri Poincar\'e}, 17(11):3177--3235, 2016.

\bibitem{0370-1298-68-10-304}
P.~G. Harper.
\newblock Single band motion of conduction electrons in a uniform magnetic
  field.
\newblock {\em Proc. Phys. Soc. A}, 68(10):874--892, 1955.

\bibitem{HatsudaKatsuraTachikawa2016}
Yasuyuki Hatsuda, Hosho Katsura, and Yuji Tachikawa.
\newblock Hofstadter's butterfly in quantum geometry.
\newblock arXiv:1606.01894, 2016.

\bibitem{HatsudaSugimotoXu2017}
Yasuyuki Hatsuda, Yuji Sugimoto, and Zhaojie Xu.
\newblock Calabi--{Y}au geometry meets electrons in 2d.
\newblock arXiv:1701.01561, 2017.

\bibitem{MR3537342}
Rinat Kashaev and Marcos Mari{\~n}o.
\newblock Operators from {M}irror {C}urves and the {Q}uantum {D}ilogarithm.
\newblock {\em Comm. Math. Phys.}, 346(3):967--994, 2016.

\bibitem{MR3546986}
Rinat Kashaev, Marcos Mari{\~n}o, and Szabolcs Zakany.
\newblock Matrix {M}odels from {O}perators and {T}opological {S}trings, 2.
\newblock {\em Ann. Henri Poincar\'e}, 17(10):2741--2781, 2016.

\bibitem{Kashani-Poor2016}
Amir-Kian Kashani-Poor.
\newblock Quantization condition from exact {WKB} for difference equations.
\newblock arXiv:1604.01690, 2016.

\bibitem{Marino2015}
Marcos Mari\~no.
\newblock Spectral theory and mirror symmetry.
\newblock arXiv:1506.07757, 2015.

\bibitem{MarinoZakany2016}
Marcos Mari\~no and Szabolcs Zakany.
\newblock Exact eigenfunctions and the open topological string.
\newblock arXiv:1606.05297, 2016.

\bibitem{MR2730782}
Nikita~A. Nekrasov and Samson~L. Shatashvili.
\newblock Quantization of integrable systems and four dimensional gauge
  theories.
\newblock In {\em X{VI}th {I}nternational {C}ongress on {M}athematical
  {P}hysics}, pages 265--289. World Sci. Publ., Hackensack, NJ, 2010.

\bibitem{MR2165903}
S.~M. Sergeev.
\newblock A quantization scheme for modular {$q$}-difference equations.
\newblock {\em Teoret. Mat. Fiz.}, 142(3):500--509, 2005.

\bibitem{MR0005741}
Andr{\'e} Weil.
\newblock {\em L'int\'egration dans les groupes topologiques et ses
  applications}.
\newblock Actual. Sci. Ind., no. 869. Hermann et Cie., Paris, 1940.
\newblock [This book has been republished by the author at Princeton, N. J.,
  1941.].

\bibitem{Zak1967}
J.~Zak.
\newblock Finite translations in solid state physics.
\newblock {\em Phys. Rev. Lett.}, 19:1385--1397, 1967.

\end{thebibliography}
%\bibliographystyle{plain}
\def\cprime{$'$} \def\cprime{$'$}

\end{document}